\documentclass[acmlarge]{acmart}

\usepackage{booktabs} 

\usepackage[ruled]{algorithm2e} 

\SetAlFnt{}
\SetAlCapFnt{}
\SetAlCapNameFnt{}
\SetAlCapHSkip{0pt}
\IncMargin{-\parindent}

\setcopyright{none}
\acmJournal{POMACS}
\acmYear{2017} \acmVolume{0} \acmNumber{0} \acmArticle{0} \acmMonth{0}
\acmArticleSeq{0}
\acmDOI{}
\acmPrice{}
\makeatletter
\renewcommand\footnotetextcopyrightpermission[1]{}
\let\@folioblob\@empty
\fancypagestyle{standardpagestyle}{%
  \fancyhf{}%
  \renewcommand{\headrulewidth}{\z@}%
  \renewcommand{\footrulewidth}{\z@}%
  \fancyhead[LE]{\@headfootfont\thepage\quad\textbullet\quad\@shortauthors}%
  \fancyhead[LO]{}%
  \fancyhead[RO]{\@headfootfont\shorttitle\quad\textbullet\quad\thepage}%
  \fancyfoot[RO,LE]{\footnotesize Manuscript submitted to \@journalName}%
}
\fancypagestyle{firstpagestyle}{%
  \fancyhf{}%
  \renewcommand{\headrulewidth}{\z@}%
  \renewcommand{\footrulewidth}{\z@}%
  \fancyfoot[RO,LE]{\footnotesize Manuscript submitted to \@journalName}%
}
\makeatother

\input{preamble.tex}

\newcommand{\N}[0]{{\bs N}}
\renewcommand{\t}[0]{{\bs \theta}}
\providecommand{\G}{}\renewcommand{\G}[0]{G_\t(\bs N)}
\usepackage{pstricks}

\newcommand{\fix}[1]{#1}

\begin{document}

\title{Accelerating Performance Inference over Closed Systems by Asymptotic Methods}
\author{Giuliano Casale}
\orcid{1234-5678-9012-3456}
\affiliation{%
  \institution{Department of Computing, Imperial College London}
  \streetaddress{180 Queen's Gate}
  \city{London}
  \postcode{SW7 2AZ}
  \country{UK}}

\date{}

\renewcommand{\shortauthors}{G. Casale}

\begin{abstract}
Recent years have seen a rapid growth of interest in exploiting monitoring data collected from enterprise applications for automated management and performance analysis. In spite of this trend, even simple performance inference problems involving queueing theoretic formulas often incur computational bottlenecks, for example upon computing likelihoods in models of batch systems. Motivated by this issue, we revisit the solution of multiclass closed queueing networks, which are popular models used to describe batch and distributed applications with parallelism constraints. We first prove that the normalizing constant of the equilibrium state probabilities of a closed model can be reformulated exactly as a multidimensional integral over the unit simplex. This gives as a by-product novel explicit expressions for the multiclass normalizing constant. We then derive a method based on cubature rules to efficiently evaluate the proposed integral form in small and medium-sized models. For large models, we propose novel asymptotic expansions and Monte Carlo sampling methods to efficiently and accurately approximate normalizing constants and likelihoods.  We illustrate the resulting accuracy gains in problems involving optimization-based inference.
\end{abstract}

\maketitle

%

\terms{Performance, Algorithms}


\section{Introduction}
During the last decade there has been a growing trend among enterprises toward exploiting large volumes of monitoring data for performance management~\cite{Gartner}. While activities such as capacity planning have been traditionally carried out by human experts, software systems to automatically forecast capacity needs are increasingly widespread in the industry. A common issue faced by these systems is automated performance model selection and parameterization, which can be dealt with using inference methods~\cite{NazP12,SpiCBK15,SutJ11,PerHK17}. 
We here focus on inference of closed queueing network models, which are often used to describe batch systems and distributed applications with parallelism limits. In such models, likelihoods can be expressed analytically if the scheduling disciplines at the resources comply with standard product-form assumptions~\cite{BasCMP75}. 

The main challenge in computing likelihoods in closed systems is to determine the normalizing constant of state probabilities, which appears explicitly in the likelihood function. Prior work has proposed methods to exactly compute normalizing constants using recursive algorithms~\cite{Buz73,Kob79,ReiK75b,ConG86}, generating functions~\cite{BerM93,HarL04}, and moment-based methods~\cite{Cas06,Cas11a}. Furthermore, methods based on Laplace transform inversion~\cite{ChoLW95}, asymptotic expansions~\cite{McKM84,Kog01}, and Monte Carlo integration~\cite{RosTW94} have led to inexpensive approximations of the normalizing constant for large models. Still, we find that maximum likelihood estimation problems remain either too expensive to solve or return largely inaccurate solutions, depending on the method used to compute the normalizing constant.

Motivated by these observations, we revisit the computational theory of product-form multiclass closed queueing networks. Our main result is to reformulate the normalizing constant in terms of an integral over the unit simplex. This development leads to novel asymptotic expansions for the normalising constant based on Laplace's method~\cite{Kass90}, which are obtained through a novel scaling that adds at every node a set of jobs that continuously loop at that node. We also obtain a novel Monte Carlo integration method, which enables the efficient sampling of the normalizing constant. Moreover, we derive novel explicit solutions for the multiclass normalizing constant in terms of algebraic sums, with time complexities similar to recursive methods such as convolution, RECAL, and mean-value analysis, but constant space requirements~\cite{Buz73,ReiK75b,ConG86,ReiL80}. We validate the effectiveness of the proposed approximations using a numerical validation involving thousands of maximum likelihood estimation problems.

The rest of the paper is organized as follows. Section~\ref{sec:background} introduces the reference model for closed systems and surveys related work. Section~\ref{sec:theory} gives novel exact theoretical results concerning the solution of closed product-form networks. Section~\ref{sec:ae} develops asymptotic expansions and Monte Carlo integration methods. Section~\ref{sec:res} investigates the accuracy of the proposed techniques. Lastly, Section~\ref{sec:conclusion} summarizes results and concludes the paper. Proofs for some technical prerequisites are given in the Appendix.

\section{Background}
\label{sec:background}

\subsection{Notation}
The reference model is a product-form closed queueing network with $M$ nodes and $R$ job classes~\cite{BasCMP75}. Let $N_r$ be the number of jobs in class $r$ and define the population vector $\bs N=(N_1,\ldots,N_R)^T$, $N=\sum_r N_r$. We assume that the first $K\leq M$ nodes have a single-server and that the remaining ones are infinite server nodes. Matrix $\t=[\theta_{kr}]$ collects the demands placed by class-$r$ jobs at node~$k$, {\em i.e.}, the product of the mean number of visits by the mean service time of the job. We denote by $\sigma_r=\sum_{k=K+1}^M \theta_{kr}$ the sum of the class-$r$ demands at all infinite servers.

Consider for example a product-form network of processor sharing and infinite server nodes having exponential service times. In this case, the model maps to a Markov process with state space
$\mathcal{S}_M=\{\bs n \in \mathbb{N}^{MR}\;|\; \textstyle n_{kr}\geq 0,\sum_{k=1}^M n_{kr}=N_r \}$, where $\bs n=(\bs n_1,\ldots,\bs n_M)$, $\bs n_k=(n_{k1},\ldots,n_{kR})$, and $n_{kr}$ is the number of class-$r$ jobs at node~$k$.  The equilibrium distribution for this process is given by~\cite{BasCMP75}
\begin{equation}
\label{eq:bcmp}
\pi(\bs n)=\frac{1}{G_\t(\N)}\prod_{i=1}^K n_i!\prod_{k=1}^M\prod_{r=1}^R \frac{\theta_{kr}^{n_{kr}}}{n_{kr}!}\qquad \bs n\in\mathcal{S}_M
 \end{equation}
where $G_\t(\N)$ is a normalizing constant over $\mathcal{S}_M$ and for vector $\bs v=(v_1,\ldots,v_n)$ we define $v=v_1+\ldots+v_n$. By the given definitions, the normalizing constant may be written as
\begin{equation}
\label{eq:Gdef}
G_\t(\N)=\sum_{\bs n\in \mathcal{S}_M} \prod_{i=1}^K n_i!\prod_{k=1}^M\prod_{r=1}^R \frac{\theta_{kr}^{n_{kr}}}{n_{kr}!}
\end{equation}
The last expression is valid for arbitrary multiclass product-form queueing networks defined in the sense of the BCMP theorem~\cite{BasCMP75}.

\subsection{Computational methods: state-of-the-art}
Since the number of states of the queueing network model grows as $\mathcal{O}(N^{MR})$ with the job population, it is usually infeasible to obtain $G_\t(\N)$ by direct summation over the state space $\mathcal{S}_M$. To tackle this issue, several computational methods have been defined over a time span of four decades. We limit here to give a high-level review, pointing to the references for details.

\subsubsection{Exact methods} The classic exact computational methods for $G_\t(\N)$ are the multiclass convolution algorithm (CA)~\cite{ReiK75b} and RECAL~\cite{ConG86}, which feature respectively $\mathcal{O}(N^R)$ and $\mathcal{O}(N^K)$ time and space requirements. Such polynomial complexities limit the application of these methods to models with a few classes or queues. Other exact algorithms with similar complexities may be found in \cite{ReiL80,BerM93,Ger95,HarL04}. The method of moments (MoM)~\cite{Cas11a} lowers the requirements approximately to $\mathcal{O}(N^2\log N)$ time and $\mathcal{O}(N\log N)$ space using a recursive system of linear equations, applicable under certain regularity conditions on $\bs \theta$. This method becomes computationally demanding as $M$ and $R$ grow simultaneously and solution times in large models are of the order of minutes, thus too expensive for optimization-based performance inference. Summarizing, several exact methods for $G_\t(\N)$ exist, but they are hardly applicable to performance inference problems due to their cost. Experiments illustrating these limitations are shown in \S\ref{sec:running}.

\subsubsection{Approximate methods} Approximate mean-value analysis (AMVA) algorithms~\cite{Sch79} provide accurate estimates of mean performance measures and are $\mathcal{O}(1)$ with respect to job populations. Yet the focus on mean performance metrics is restrictive since inference problems typically require a probability model such as \eqref{eq:bcmp}, for example to express prior distributions on parameters or to infer an optimal parameterization using likelihood maximization~\cite{Paw01,SpiCBK15}. AMVA does not apply to these problems as it neither computes likelihoods nor probabilities.

Flow-equivalent methods alternatively aggregate a subnetwork into a node with state-dependent service rates, which may be solved for state probabilities~\cite{ChaHW75b}. Unfortunately, this method is normally too expensive to apply in multiclass models, where the parameterization of the flow-equivalent server requires to determine service rates under all possible combinations of jobs residing at the node.

Such limitations are addressed by specialized approximations for $G_\t(\N)$, which include Monte Carlo integration, numerical methods, and asymptotic expansions. Monte Carlo integration methods are first introduced in \cite{RosTW94}, based on the following integral form \cite{McKM84}
\begin{equation}
\label{eq:intform}
\G=\frac{1}{N_1!\cdots N_R!} \int_{ \mathbb{R}^{K}_+} e^{-y}\prod_{r=1}^R\left(\sigma_r+\sum_{k=1}^K\theta_{kr}y_k\right)^{N_r} d\bs y
\end{equation}
where $\mathbb{R}^{K}_+=\{\bs y\in \mathbb{R}^K|\bs y\geq\bs 0\}$, with $\bs 0=(0,\ldots,0)^T$ and $y=\sum_{k} y_k$. Expression \eqref{eq:intform} is obtained by expressing the $n_i!$ terms in \eqref{eq:bcmp} using the integral form of the gamma function and subsequently by repeated application of the multinomial theorem
\begin{equation}
\label{eq:multinomial}
\left(\sum_{k=1}^K x_k\right)^V = V! \sum_{\bs v\geq \bs 0\atop v=V}\prod_{k=1}^K \frac{x_k^{v_k}}{v_k!}
\end{equation}
where $\bs v\in \mathbb{N}^K$. The integral form \eqref{eq:intform} can be efficiently evaluated using importance sampling~\cite{RosTW94}, normally requiring $J=10^5-10^7$ samples to approximate $G_\t(\N)$ with low variance. Computing millions of samples is acceptable for evaluating individual models, but places an excessive overhead for use in optimization-based inference. Moreover, the variance of the Monte Carlo estimators for $\G$ can adversely affect the identification of the search direction~\cite{WanCS16}.

Numerical methods for $G_\t(\N)$ include Laplace transform inversion (LTI)~\cite{ChoLW95} and ODE-based methods based on Taylor expansion (TE)~\cite{WanCS16}. LTI allows for arbitrary approximation accuracy, but can still incur a significant computational cost. For instance, \cite[p. 967]{ChoLW95} provides an example where LTI requires $\approx 10^{12}$ operations on a model with $K=64$ and $R=9$, which is beyond the acceptable cost for a single iteration of an optimization program. Instead, the approximations proposed in this paper scale efficiently to models of this size. TE is theoretically $\mathcal{O}(1)$, but 
it becomes difficult to apply in large models due to the rapid growth of $\G$ that affects numerical precision.

\subsubsection{Asymptotic expansions}
Asymptotic expansions for $G_\t(\N)$ are $\mathcal{O}(1)$ and thus capable of accelerating optimization. Expansions appear in \cite{McKM84,KneT92,Kog94,KogY96,Kog01} and are discussed below.  Other asymptotic methods exist but they are not relevant to the present work as they either focus on single-class models only~\cite{GeoXS12} or study asymptotic values of mean-value performance metrics \cite{AnsC08, BalS97, BerBK99, KelMW09, KneT98}, whereas we focus here on computing likelihoods in multiclass systems.

PANACEA (PAN)~\cite{McKM84} is applicable only to models with infinite servers and in normal usage, {\em i.e.}, where resources are lightly utilized so that $\max_i \alpha_i < 1$, where $\alpha_i=\sum_r N_r \theta_{ir} \sigma^{-1}_r$, $i=1,\ldots,K$. Normal usage conditions tend to be restrictive in applications, where the analysis of heavy-load regimes is of practical importance.

The ray method~\cite{KneT92} (RAY) is an approximation method for PDEs. Combined with singular perturbation theory, RAY provides an approximation for $\G$ under an increasing number of nodes and a simultaneous scaling of their service demands. We extensively compare our results against the baseline provided by this method. 

Saddle-point approximation (SPA) provides asymptotic expansions of contour integrals arising from the generating function of $\G$. A limitation of SPA is that explicit formulas are available only for small models, although it principle the method may be generalized~\cite{Kog94,KogY96,Kog01}.

As we show in \S\ref{sec:ae}, our asymptotic expansions are based on Laplace's method~\cite{Kass90}, which may be seen as a specialization of the saddle-point method for real integrals. This substantially differs from the SPA method which applies to contour integrals in the complex domain and leads to rather different expressions for $G_\t(\N)$.

\begin{table}[t]
\caption{Demand estimation results. {\normalfont $T$ = timeout (10 min). Runtimes are rounded up to the nearest integer. On this example, memory requirements are negligible for all methods.}}
{
\begin{center}
\begin{tabular}{lcccccc}
\toprule
$M=K=3$, $R=3$ & \multicolumn{3}{c}{{ Abs. Perc. Error (\%)}} & \multicolumn{3}{c}{{ Time (seconds)}}\\
\midrule
 $N_r=N/R=$ & {\em 2} & {\em 20} & {\em 40} & {\em 2} & {\em 20} & {\em 40} \\
\midrule
 CA (No timeout) & 0.3 & 0.0 & 0.0 & 2 & 419 & 3413 \\
\midrule
 CA & 0.3 & 0.0 &  40.3 & 2 & 419 & $T$ \\
 MCI3 &  90.9 & 138.9  & 115.1 & 5 & 5 & 7\\
 MCI6 & 135.9 & 91.1 & 118.3 & 185 & 171 & 137 \\
 MoM & 15.4 & 51.2 & 71.0 & $T$ & $T$ & $T$  \\
 NOG & 38.6 & 92.2 & 96.1 & $1$ & $1$ & $8$ \\
 RAY & 82.1 & 72.8 & 59.4 & 6 & 3 & 8  \\
 RECAL & 0.3 & 31.7 & 76.5 & 3 & $T$ & $T$ \\
 TE-2 & 38.7 & 92.2 & 96.1 & 143 & 174 &  138 \\
 TE-3 & 35.2 & 91.3 & 95.6 & $T$ & $T$ & $T$ \\
\bottomrule
\end{tabular}
\end{center}
}
\label{tab:motivating}
\end{table}

\subsection{Motivating example: demand estimation from state samples}
\label{sec:running}
To illustrate the limitations of existing computational techniques, we compare prior art methods in a likelihood maximization application. Assume to measure a set of $L$ state samples, $\bs n^{(l)}\in\mathcal{S}_M$, $l=1,\ldots,L$. We seek for a maximum likelihood estimator (MLE) for the demand matrix ${\bs \theta}$. In practice, problems of this kind arise during model selection and calibration, where one seeks for an optimal parameterization and the $\bs n^{(l)}$ samples represent system state measurements. Likelihood-based estimation offers a number of advantages over other estimation techniques, for example it can cope with missing and aggregate data~\cite{Paw01}.

We assume $\sigma_r=0$ and knowledge of $\theta_r=\sum_k \theta_{kr}$, $\forall r$, {\em i.e.}, the end-to-end response time of a single class-$r$ request when $N=1$. From \eqref{eq:bcmp} the log-likelihood of $\t$ is given by
\begin{equation}
\label{eq:llog}
 \mathcal{L}(\bs \theta) = \sum_{l=1}^L \log C(\bs n^{(l)})+L\sum_{k}\sum_{r} \widetilde{Q}_{kr}\log \theta_{kr}-L\log \G
\end{equation}
where \fix{$C(\bs n)=\prod_{i=1}^K n_i!\prod_{k=1}^M\prod_{r=1}^R (n_{kr}!)^{-1}$, and} $\widetilde{Q}_{kr}=\sum_{l} n^{(l)}_{kr}/L$ is the measured mean queue-length of class $r$ at node $k$. The first term can be neglected upon optimizing over $\t\geq \bs 0$. The cost of computing $\mathcal{L}(\bs \t)$ is thus dominated by the cost of determining $\log \G$.

We consider \eqref{eq:llog} for a model with $M=K=3$, $R=3$, $\bs N=(N,N,N)/3$, and $\theta_{kr}=k*r$ and seek for a (local) MLE that maximizes $\mathcal{L}(\bs \theta)/L$ subject to $\t\geq \bs 0$. We also let $L\to \infty$ by using in place of the $\widetilde{Q}_{kr}$ the exact mean queue-lengths computed by mean-value analysis~\cite{ReiL80}. We apply MATLAB's {\sf\small fmincon} interior point algorithm, and calculate $\G$ at each iteration using one among CA, RECAL, MoM, TE, or RAY. We also use Monte Carlo integration (MCI) with AMVA-based initialization \cite{WanCS16}. The assumptions underpinning PAN and SPA are not met on this example: PAN requires infinite servers; SPA is not available for models with $3$ nodes and $3$ classes or larger. TE also requires infinite servers, but we can set $\sigma_r=10^{-8}$, $\forall r$; a similar perturbation cannot be used with PAN since the method also requires normal usage. Lastly, we include in the experiment a variant of \eqref{eq:llog} where we neglect the normalizing constant by setting $\log \G=0$. This variant is denoted by NOG and corresponds to the log-likelihood formula for a product-form open queueing network with demands $\t$. Each chosen method is initialized at the same point, sampled from a uniform distribution. We set a timeout of $10$ minutes, after which {\sf\small fmincon} returns after completing the running iteration.
Experiments are run on a quad-core desktop computer.

Table~\ref{tab:motivating} shows execution times with $N_r=N/R=2,20,40$ jobs and the absolute percentage errors of the returned demands with respect to the true $\t$. The suffixes for MCI and TE are the number of samples $J$ and the scale of the ODE step size~$\tau$, respectively, {\em e.g.}, MCI3 has $J=10^{3}$ and TE-2 has $\tau=10^{-2}$.
The {\em CA (No timeout)} method is exact and thus provides an upper bound on achievable accuracy on this instance. For this method, we run the optimization until termination, computing the normalizing constant at each iteration using CA. This baseline is required since the problem \eqref{eq:llog} is non-convex, thus the choice of the initial point affects the achievable accuracy and it is thus undesirable to reason on absolute error alone.

We note that all methods incur a considerable degradation of accuracy and running times as the population $N$ grows. 
Some methods, such as TE, have a similar (or worse) performance than NOG, which ignores the normalizing constant.
CA is the best among the exact methods, but its execution times grow quickly and on larger models become infeasible. Among existing approximations, the RAY asymptotic expansion achieves the best results, although the errors remain high, around 59\%-82\%. However, computational times are scalable. A similar conclusion applies to MCI with a small number of samples. This motivates us to further investigate into asymptotic expansions and Monte Carlo integration methods.
We also remark that on this example the expansion proposed later in \S\ref{sec:ae} achieves less than $0.7\%$ absolute percentage error in all the three cases, with runtimes between $2s$ and $4s$. {A validation on a broader set of instances is presented in \S\ref{sec:res}.}

\section{Exact results}
\label{sec:theory}
In order to inexpensively approximate the normalizing constant, we first derive novel integral expressions for $\G$. This derivation leads to novel numerical approximations and provides a theoretical baseline for developing asymptotic results.

\subsection{Integral form over the unit simplex}
We first derive an exact integral form for the normalizing constant in networks without infinite servers.
\begin{theorem}
\label{thm:nois}
In a multiclass closed queueing network with $K$ single-server nodes
\begin{equation}
\label{eq:Gsimplex}
\G=  \frac{(N+K-1)!}{N_1! \cdots N_R!} \int_{\Delta_{K}} \prod_{r=1}^R \Biggl(\sum_{k=1}^{K} \theta_{kr}u_k\Biggr)^{N_r} d\bs u
\end{equation}
where $\Delta_{K}=\{\bs u\in \mathbb{R}^K\,|\, u_i\geq 0, \sum_{i=1}^K u_i=1\}$ is the unit simplex.
\end{theorem}
\begin{proof}
The multinomial theorem~\eqref{eq:multinomial} implies that for any set of real numbers $(a_1,\ldots,a_R)$ and variables $\bs t=(t_1, \mathellipsis,t_R)^T$ we can write
\begin{align}
 \label{eq:prodsum2}
 a^{N_1}_1 a^{N_2}_2 \cdots a^{N_R}_R=\frac{1}{N!} \frac{\partial^{N_1}\cdots \partial^{N_R}}{\partial t_1^{N_1}\cdots \partial t_R^{N_R}}\left(\sum_{i=1}^R a_i t_i \right)^N
 \end{align}
Since $\sigma_r=0$, we apply \eqref{eq:prodsum2} to the product in the integrand of \eqref{eq:intform} with $a_r=\sum_{k}\theta_{kr}y_k$, finding after exchanging the order of differentiation and integration
\begin{align}
\G&=\frac{1}{N_1!\cdots N_R!} \frac{\partial^{N_1}\cdots \partial^{N_R}}{\partial t_1^{N_1}\cdots \partial t_R^{N_R}} \int_{ \mathbb{R}^K_+} \frac{e^{-y}}{N!}  \left(\sum_{k=1}^{K} {\theta}_{k}y_k\right)^{N}d\bs y\nonumber\\
&=\frac{1}{N_1!\cdots N_R!} \frac{\partial^{N_1}\cdots \partial^{N_R}}{\partial t_1^{N_1}\cdots \partial t_R^{N_R}} g_{\t\bs t}(N)\label{eq:Gtn}
\end{align}
where $g_{\t\bs t}(N)$ is the normalizing constant of a single-class model with demands $\theta_k=\sum_r\theta_{kr}t_r$, $k=1,\ldots,K$, and the last passage follows by \eqref{eq:intform}.
We then prove in Appendix~\ref{apx:L1} the following equivalence
\begin{equation}
\label{eq:dd}
g_{\t\bs t}(N)=[\theta_1,\mathellipsis,\theta_{K}]x^{N+K-1}
\end{equation}
where the right-hand side is the divided difference\footnote{For a given function $f(x)$, divided differences extend the notion of forward difference of $f(x)$ to a set of interpolation points arbitrarily located in the domain of $f(x)$. We point to \cite{Atk89,Mil33} for further details.} of $x^{N+K-1}$ relatively to the interpolation points $\theta_1$,\ldots, $\theta_{K}$. 
This expression is valid for single-class normalizing constants with arbitrary demands.
We can then apply to the last expression the Hermite-Genocchi formula~\cite{Atk89}, which is a classic integral form for divided differences
\begin{equation}
\label{eq:genocchi}
  [\theta_1,\mathellipsis,\theta_{K}]f(x)=\int_{\Delta_{K}} f^{(K-1)}(\theta_1u_1+\mathellipsis+\theta_{K}u_{K}) d\bs u
\end{equation}
where $f^{(K-1)}(x)$ is the $(K-1)$th derivative of $f(x)$ and $\Delta_{K}$ is the unit simplex. Here we set $f(x)=x^{N+K-1}$ and show in Appendix~\ref{apx:L2} that, for this specific choice of $f(x)$, \eqref{eq:genocchi} also holds under nondistinct $\theta_i$, a case normally not covered by the Hermite-Genocchi formula.
Using \eqref{eq:dd} in \eqref{eq:Gtn}, followed by \eqref{eq:genocchi}, we get
\begin{equation*}
\G=\frac{(N+K-1)!}{N_1!\cdots N_R!}\int_{\Delta_{K}} \frac{1}{N!} \frac{\partial^{N_1}\cdots \partial^{N_R}}{\partial t_1^{N_1}\cdots \partial t_R^{N_R}} \left(\sum_{k=1}^{K} \theta_k u_k\right)^{N}d\bs u
\end{equation*}
Recalling that $\theta_k=\sum_r\theta_{kr}t_r$ and applying \eqref{eq:prodsum2} to the integrand, we find \eqref{eq:Gsimplex}.
\end{proof}
{
Theorem~\ref{thm:nois} provides a novel integral form for the multiclass normalizing constant, with an integrand similar to \eqref{eq:intform}, but defined over a bounded domain. It is also possible to verify that \eqref{eq:intform} follows from \eqref{eq:Gsimplex} using the Laplace transform, once the integration domain is reformulated in a suitable parametric form~\cite{LasZ01}.

It is useful to note that a shorter proof of Theorem~\ref{thm:nois} follows by first applying the multinomial theorem \eqref{eq:multinomial} to each factor in the integrand of \eqref{eq:Gsimplex} and then using term-by-term the Dirichlet integral
\begin{equation}
\label{eq:dirichlet}
\int_{\Delta_{K}} \prod_{k=1}^K u_k^{n_k} d\bs u = \dfrac{\prod_{k=1}^K n_k!}{(\sum_{k=1}^K n_k+K-1)!}
\end{equation}
This yields \eqref{eq:Gdef} after noting that in the absence of infinite server nodes it is $\sum_{k=1}^K n_k=N$. Compared to this simple derivation, the proof of Theorem~\ref{thm:nois} introduces \eqref{eq:prodsum2}, which is used in the next section to obtain explicit solutions. Similar mappings between sums and products are important also in multivariate statistical analysis~\cite{Kan07}. Moreover, the proof of Theorem~\ref{thm:nois} shows that multiclass normalizing constants may be expressed as derivatives of single-class normalizing constants, and that the latter may be seen as divided differences of the power function.
}

\subsection{Explicit solutions}
While our interest is on deriving approximations, novel exact computational formulas may also be obtained from Theorem~\ref{thm:nois}.
Such expressions are not used throughout due to their cost, but they appear of theoretical interest due to the lack of similar expressions for the multiclass normalizing constant.

\begin{corollary} \label{cor:Gps}
The normalizing constant of a closed multiclass queueing network is given by
\begin{equation}
\label{eq:Gps}
\G = \sum_{\bs 0\leq \bs t\leq \bs N} \frac{(-1)^{N-t}}{N_1!\cdots N_R!}\prod_{r=1}^R{N_r \choose t_r} g_{\t\bs t}(N)
\end{equation}
where $\bs t=(t_1,\ldots,t_R)^T$, $t=\sum_r t_r$, and $g_{\t\bs t}(N)$ is the normalizing constant of a single-class model with demands $\theta_k=\sum_{r=1}^R t_r \theta_{kr}$.
\end{corollary}
\begin{proof}
We consider finite differences~\cite{Mil33}, where \eqref{eq:prodsum2} is replaced by~\cite{Bal11}
\begin{equation}
 \label{eq:prodsum}
 a^{N_1}_1 a^{N_2}_2 \cdots a^{N_R}_R=\sum_{\bs 0\leq \bs t\leq \bs N} \frac{(-1)^{N-t}}{N!} \prod_{s=1}^R {N_s\choose t_s} \left(\sum_{r=1}^R t_r a_r\right)^N
 \end{equation}
The result follows by applying \eqref{eq:prodsum} to the product in the integrand of \eqref{eq:Gsimplex} and recognizing $g_{\t\bs t}(N)$ in the resulting expression.
\end{proof}
Computing $g_{\t\bs t}(N)$ explicitly using the closed-form formulas in~\cite[Eq. 3.12]{BerM93} implies for \eqref{eq:Gps} a theoretical complexity of $\mathcal{O}(N^R)$ time and $\mathcal{O}(1)$ space. For example, in the special case where demands are distinct, the normalizing constant can be computed in $\mathcal{O}(1)$ as~\cite{Koe58,BerM93}
\begin{equation}
\label{eq:exsingle}
g_{\t\bs t}(N)=\sum_{k=1}^K \frac{\theta_k^{N+K-1}}{\prod_{i\neq k} (\theta_k-\theta_i)}
\end{equation}
yielding by \eqref{eq:Gps} the following explicit expression for the normalizing constant
\begin{equation}
\G = \sum_{\bs 0\leq \bs t\leq \bs N} \frac{(-1)^{N-t}}{N_1!\cdots N_R!}\prod_{r=1}^R{N_r \choose t_r} \sum_{k=1}^K \dfrac{(\sum_{s=1}^Rt_s\theta_{ks})^{N+K-1}}{\prod_{i\neq k} (\sum_{s=1}^Rt_s(\theta_{ks}-\theta_{is}))}
\end{equation}
\fix{Two provisions apply to the last expression. First, the convention $0/0=0$ is used. Second, the expression requires the induced demands $\theta_k=\sum_{s=1}^R t_s\theta_{ks}$ to be pairwise distinct for every $\bs t$ in the summation. Otherwise the following result must be considered.}
A similar formula holds for the general case if one uses \cite[Eq. 3.12]{BerM93} in place of \eqref{eq:exsingle}. Consider the single-class demands $\theta_k=\sum_{r=1}^Rt_r \theta_{kr}$. Assume that $\theta_k$ has multiplicity $m_k\geq 1$ and let $K'$ be the number of distinct demands. Plugging \cite[Eq. 3.12]{BerM93} into \eqref{eq:Gps} yields the general expression
\begin{multline}
\G = \sum_{\bs 0\leq \bs t\leq \bs N} \frac{(-1)^{N-t}}{N_1!\cdots N_R!}\prod_{s=1}^R{N_s \choose t_s} \sum_{j=1}^{K'} (-1)^{m_j-1} (\textstyle\sum_{s=1}^Rt_s\theta_{js})^{N+M-m_j}\\
\times\sum_{\substack{\bs r\geq 0\\r=m_j-1}} (-1)^{r_j} {N+r_j \choose r_j} \prod_{k=1\atop k\neq j}^{K'} {m_k+r_k-1\choose r_k} \frac{(\sum_{s=1}^Rt_s\theta_{ks})^{r_k}}{(\sum_{s=1}^Rt_s(\theta_{js}-\theta_{ks}))^{m_k+r_k}}
\end{multline}
where $\bs r\in \mathbb{N}^{K'}$, $r=\sum_{j=1}^{K'} r_j$. We are also in condition to derive another explicit formula for $\G$.
\begin{corollary} \label{cor:Gms}
The normalizing constant of a closed multiclass queueing network model can be expressed as
\begin{equation}
\label{eq:Gms}
\G= \sum_{\substack{\bs h\geq \bs 0:\\h\leq N}} \frac{(-1)^{N-h}}{N_1!\cdots N_R!} {N+K-1\choose N-h} \prod_{r=1}^R \left(\sum_{i=1}^K h_i\theta_{ir}\right)^{N_r}
\end{equation}
where $\bs h\in \mathbb{N}^K$ and $h=\sum_{i} h_i$.
\end{corollary}
\begin{proof}
Observe that the specialization of \eqref{eq:bcmp} to single-class models is
\begin{equation}
\label{eq:gdef}
g_{\t\bs t}(N)=\sum_{\bs m\in \mathcal{S}_K} \prod_{k=1}^K \theta^{m_k}_k
\end{equation}
Applying \eqref{eq:prodsum} to \eqref{eq:gdef} with $a_k=\theta_k$, and using the definition of $\mathcal{S}_K$ yields
\[
 g_{\t\bs t}(N) = \sum_{\substack{\bs m\geq \bs 0:\\ m=N}} \sum_{\bs 0\leq \bs h\leq \bs m} \fix{\frac{(-1)^{N-h}}{N!}}\prod_{j=1}^K {m_j \choose h_j}\left(\sum_{i=1}^K \sum_{r=1}^R h_i t_r \theta_{ir}\right)^{N}
\]
with $h=\sum_{i} h_i$.
Plugging the last formula into \eqref{eq:Gps} and using \eqref{eq:prodsum} to eliminate the dependence on~$\bs t$ gives after rearranging terms
\[
\label{eq:intermediate2}
 \G=\sum_{\substack{\bs m\geq \bs 0:\\ m=N}}  \sum_{\bs 0\leq \bs h \leq \bs m}\frac{(-1)^{N-h} }{N_1!\cdots N_R!} \prod_{j=1}^K  {m_j \choose h_j} \prod_{r=1}^R {\left(\sum_{i=1}^K h_i\theta_{ir}\right)^{N_r}}
\]
We can here use a single summation on $\bs h\geq \bs 0$, $h\leq N$, after noting that
\begin{equation}
\label{eq:prodbin}
\sum_{\substack{\bs m\geq \bs h:\\ m=N}} \prod_{i=1}^K  {m_i \choose h_i}= {N+K-1\choose N-h}
\end{equation}
This identity can be proved by first rewriting the expression in terms of $\bs v=\bs m-\bs h\geq \bs 0$ and
then iteratively applying a corollary of Vandermonde's convolution~\cite[Eq. 3]{Gou56}.
\end{proof}
This explicit form requires $\mathcal{O}(N^K)$ time and $\mathcal{O}(1)$ space, which makes it preferable to \eqref{eq:Gps} on models with many classes, but a small number of nodes. To the best of our knowledge, \eqref{eq:Gps} and \eqref{eq:Gms} are the only exact and tractable algebraic expressions for $\G$ with a $\mathcal{O}(1)$ space requirement. This improves over the space requirements of recursive algorithms such as CA and RECAL, while retaining the same time complexities, which may be useful in cases where one wants to solve several models in parallel without incurring into memory bottlenecks. However, in practice the above expressions are applicable only to models where $H=\min(K,R)$ is not too large ({\em e.g.}, $H\leq 4$), typically with up to a few tens of jobs. Moreover, due to the large magnitude of the terms, multi-precision arithmetic should be used to avoid numerical issues upon computing \eqref{eq:Gps} and \eqref{eq:Gms}. The techniques developed later do not suffer these problems and can help to approximate larger models.


\subsection{Infinite server nodes}
Consider now a model where the first $K$ nodes are single-server queues and the remaining $M-K$ nodes are infinite servers. The following corollary generalizes the integral form.
\begin{corollary}
\label{thm:is}
In a model including infinite server nodes
\begin{equation}
\label{eq:GZsimplex}
\G= \int_{\Delta_{K}}  \int_{v=0}^{+\infty} \frac{e^{-v} v^{K-1}}{N_1!\cdots N_R!} \prod_{r=1}^R \Biggl(\sigma_r+v\sum_{k=1}^{K} \theta_{kr}u_k\Biggr)^{N_r} dv\,d\bs u
\end{equation}
\end{corollary}
\begin{proof}
{
We plug the definition of $\sigma_r=\fix{\sum_{k=K+1}^{M} \theta_{kr}}$ in \eqref{eq:GZsimplex} and
use the multinomial theorem~\eqref{eq:multinomial} to write
\[
\G=\sum_{\bs n\in \mathcal{ S}_M} \int_{\Delta_{K}}\int_{v=0}^{+\infty}{e^{-v} v^{n+K-1}} \prod_{i=1}^M\prod_{r=1}^R \frac{\theta_{ir}^{n_{ir}}}{n_{ir}!} \prod_{k=1}^K u_k^{n_k} dv\,d\bs u
\]
where $n=\sum_{k=1}^K n_k$ and $n_k=\sum_{r} n_{kr}$. Noting that $\int_{v=0}^{+\infty} {e^{-v} v^{n+K-1}}dv=(n+K-1)!$, we obtain \eqref{eq:Gdef} by \eqref{eq:dirichlet}.
}
\end{proof}
{
}

\subsection{Numerical evaluation}
{
Cubature rules are interpolation formulas that approximate a multidimensional integral by computing the integrand at a finite set of points~\cite{Coo03}. For polynomial integrands, cubature rules may also allow the exact evaluation of the integral, if interpolation occurs at a large enough set of points. An advantage of \eqref{eq:Gsimplex} over \eqref{eq:intform} is that it expresses $\G$ as an integral of a polynomial over the simplex, making it suitable for application of cubature rules.

We focus here on Grundmann-M\"oller (GM) cubature rules, which are tailored to the exact and approximate integration of polynomials over the simplex~\cite{GruM78}.} Applying directly the definition of GM cubature rule of degree $2S+1$ to \eqref{eq:Gsimplex} leads to the following expression~\cite{GruM78} 
\begin{equation}
 \label{eq:grnmol}
\G = \frac{(N+K-1)!}{N_1!\cdots N_R!}\sum_{i=0}^S w_i \sum_{\substack{\bs b\geq \bs 0:~\\b=S-i}} \prod_{r=1}^R\Biggl(\sum_{j=1}^K\frac{(2b_j+1)\theta_{jr}}{(2S+K-2i)}\Biggr)^{N_r}
\end{equation}
where $\bs b\in \mathbb{N}^K$, $b=\sum_{i=1}^K b_i$, and the weights are
 \[
 w_i = (-1)^i2^{-2S}\frac{(2S+K-2i)^{2S+1}}{i!(2S+K-i)!}
 \]
The number of points in the rule~\eqref{eq:grnmol} is $L={K+S\choose S}$, thus worst-case complexity is $\mathcal{O}(S^K)$ time and $\mathcal{O}(1)$ space.

\begin{figure}[t!]
\centering
\subfigure [$S=8$] {
 {\resizebox{0.3\linewidth}{!}{\includegraphics[]{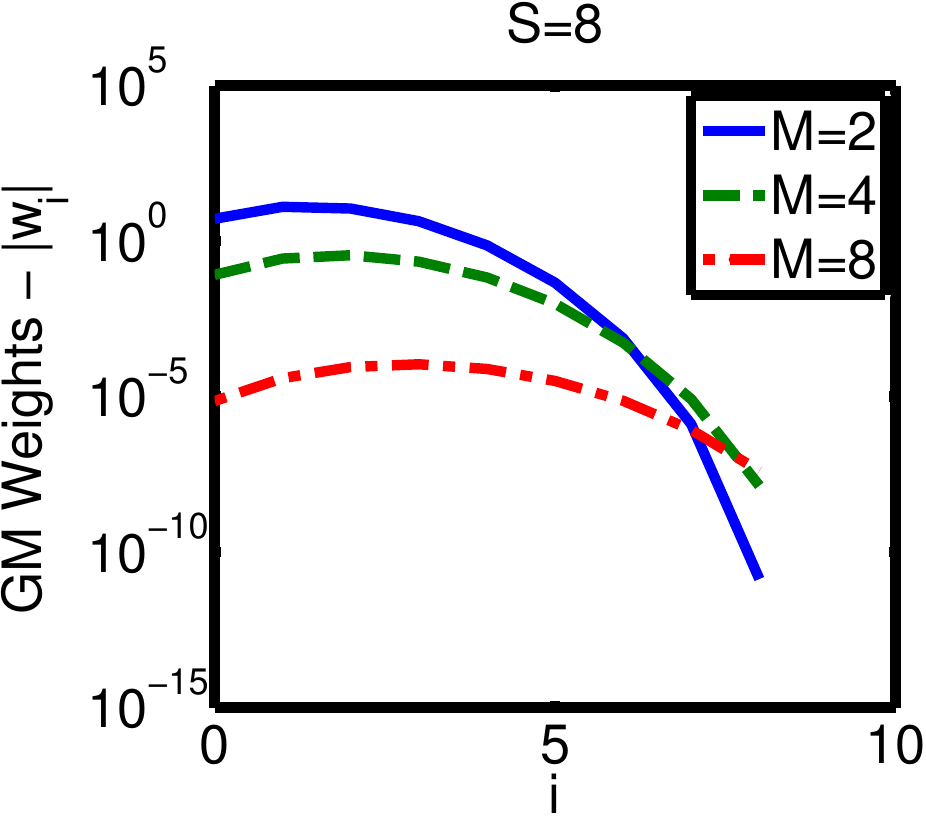}}}
 }\qquad
\subfigure [$S=16$] {
 {\resizebox{0.3\linewidth}{!}{\includegraphics[]{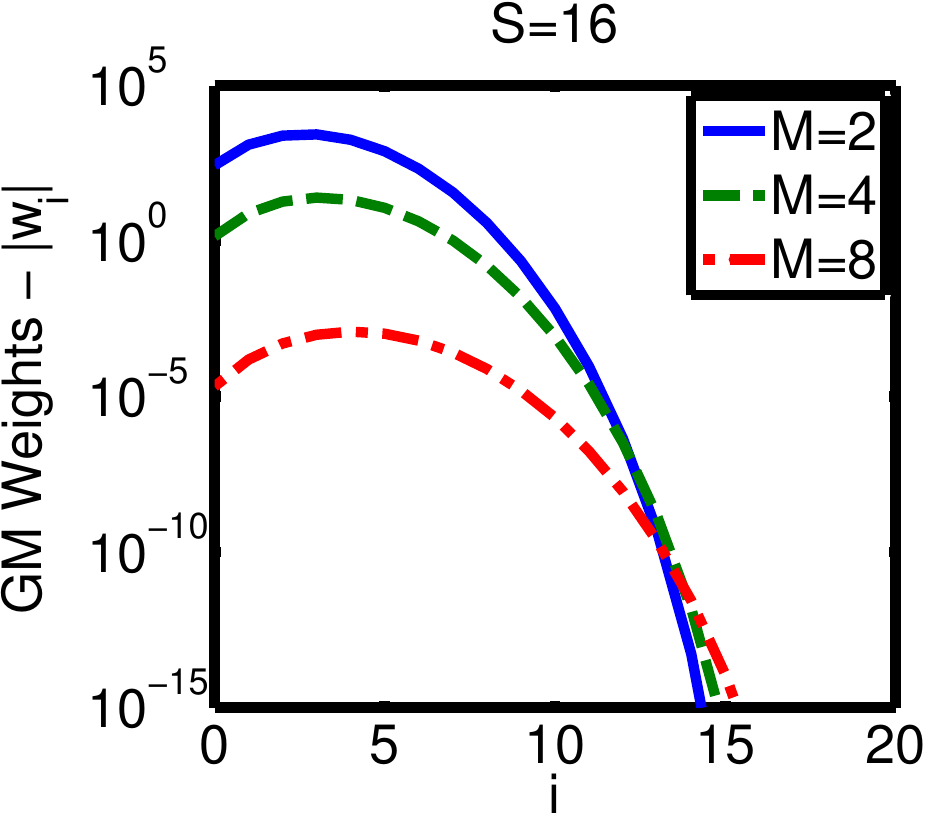}}}
 }
 \caption{Grundmann-M\"oller weights $w_i$ for models with $M=K$ single-server nodes}
 \label{fig:weights}
\end{figure}

If the integrand is a multivariate polynomial of degree $N$, then a GM cubature rule of degree $S=\lceil(N-1)/2\rceil$ returns the exact value of the integral in $\mathcal{O}((N/2)^K)$ time and $\mathcal{O}(1)$ space~\cite{GruM78}. This is indeed the case for both \eqref{eq:Gsimplex} and \eqref{eq:GZsimplex}. In the case without infinite servers, this is evident since the integrand is a product of linear forms. We now show that the same conclusion holds in models with infinite servers. Let ${G^{\bs u}_\t(\bs N)}$ be the normalizing constant for a model composed of an infinite server node with demand~$\sigma_r$ and $K$ identical single-server nodes having class-$r$ demand $\tilde{\theta_{r}}=\sum_k \theta_{kr}u_k$. We have the following result.
\begin{proposition}
\label{pro:moment-partition}
$$\int_{v=0}^{+\infty}\frac{e^{-v} v^{K-1}}{N_1!\cdots N_R!} \prod_{r=1}^R \Biggl(\sigma_r+v\sum_{k=1}^K \theta_{kr}u_k\Biggr)^{N_r} dv=\Gamma(K){G^{\bs u}_\t(\bs N)}$$
\end{proposition}
\begin{proof}{
Let $\bs n\in \mathbb{N}^R$, $n=\sum_r n_r$. Using \cite[Thm. 2]{Cas11a} we get
\[
{G^{\bs u}_\t(\bs N)}=\sum_{\bs 0\leq \bs n\leq \bs N}{n+K-1 \choose n} n!\prod_{r=1}^R \frac{\tilde{\theta}^{n_{r}}_{r}}{n_{r}!}\cdot\frac{\sigma_{r}^{N_r-n_{r}}}{(N_r-n_{r})!} 
=\sum_{\bs 0\leq \bs n\leq \bs N}\frac{\Gamma{(n+K)}}{\Gamma(K)}\prod_{r=1}^R \frac{\tilde{\theta}^{n_{r}}_{r}}{n_{r}!}\cdot\frac{\sigma_{r}^{N_r-n_{r}}}{(N_r-n_{r})!}
\]
We now plug the integral expression of $\Gamma{(n+K)}$ and the statement follows by the multinomial theorem~\eqref{eq:multinomial} and by definition of $\tilde{\theta_{r}}$.
}
\end{proof}
Since ${G^{\bs u}_\t(\bs N)}$ is a normalizing constant with demands $\tilde{\theta}_r$ that are linear functions of $\bs u$, by definition it is a multivariate polynomial of degree $N$ in $\bs u$. The theorem thus confirms that the inner integral in \eqref{eq:GZsimplex} is a polynomial of degree $N$. Thus a GM rule with $S=\lceil(N-1)/2\rceil$ returns the exact value of $\G$ also in the presence of infinite servers. Therefore, similarly to \eqref{eq:Gps} and \eqref{eq:Gms}, \eqref{eq:grnmol} provides an exact expression for $\G$ that is both explicit and tractable.

Using smaller values of $S$ trades accuracy for speed, as it is possible to approximate $\G$ by truncation of the outer summation of \eqref{eq:grnmol}. This is an effective procedure thanks to the rapid decay of the weights $w_i$, as illustrated in Figure \ref{fig:weights}. For large enough $i$, the weights quickly and monotonically decrease, thus a few outer iterations of \eqref{eq:grnmol} are sufficient to return a good approximation. As we show later, GM rules perform very well on small and medium-sized inference problems. However, as the model size grows, the asymptotic expansions introduced in the next sections are normally more efficient.

\section{Asymptotic Expansion}
\label{sec:ae}

\subsection{Preliminaries}
We now exploit the geometry of the unit simplex to derive an asymptotic expansion for $\G$. Our approach first applies a logistic transformation to the integration variables in its integral form~\cite{Ait82}. This is a classic method to map integrands defined over the $n$-dimensional simplex to $\mathbb{R}^{n}$. We find after this transformation that, as $N$ grows, the integrand becomes increasingly peaked at a unique point in the interior of the integration domain, satisfying the conditions for Laplace's method~\cite{Kass90}.

{
A technical requirement for our argument to hold is that all queue-lengths grow asymptotically large as $N\to +\infty$, a property which is violated by non-bottleneck nodes. To address this issue, we introduce a novel scaling where we also slowly increase at every node a population of jobs that permanently reside at the node itself, perpetually self-looping. That is, we introduce $K$ additional classes, each with population $\epsilon N$, $\epsilon>0$, where the $i$th class is composed by jobs that self-loop at node $i$, placing a unit service demand at each visit. In this way we are considering the perturbed normalizing constant
\begin{equation}
G_\t^\epsilon(\bs N)=\fix{\frac{1}{N_1!\cdots N_R!((\epsilon N)!)^K}} \int_{v=0}^{+\infty}  e^{-v}v^{K(1+\epsilon N)-1} \int_{\Delta_{K}}\prod_{i=1}^K u_i^{\epsilon N} \prod_{r=1}^R \Biggl(\sigma_r+v\sum_{k=1}^{K} \theta_{kr}u_k\Biggr)^{N_r} d\bs u\,dv
\end{equation}
where $\eta_\epsilon(N)=N+K(1+\epsilon N)$ and $\lim_{\epsilon\to 0^+} G^{\epsilon}_\t(\bs N)=\G$.}


Our asymptotic expansion is proved for an auxiliary function $I_{\epsilon}(N)$, which uniquely defines $G_\t^\epsilon(\bs N)$. This function is defined as follows. First, we allow for real values of $\epsilon$ by expressing where needed factorials in $G_\t^\epsilon(\bs N)$ using the gamma function $\Gamma(\cdot)$. Without loss of generality, we then normalize the service demands in the auxiliary function to range in $[0,1]$. That is, we set
\begin{equation}
\label{eq:Geps}
G^{\epsilon}_\t(\bs N)=\frac{\Gamma(\eta_\epsilon(N))}{N_1!\cdots N_R!(\Gamma(1+\epsilon N))^K}\,{I_{\epsilon}(N)}\,\prod_{r=1}^R \alpha_r^{N_r}
\end{equation}
where $\alpha_r=\sigma_r + \max_{i} \theta_{ir}$ and we define the auxiliary function
\begin{equation}
\label{eq:Idef}
I_{\epsilon}(N)= \frac{1}{\Gamma(\eta_\epsilon(N))} \int_{v=0}^{+\infty}  e^{-v}v^{K(1+\epsilon N)-1}  \int_{\Delta_{K}} \prod_{i=1}^K u_i^{\epsilon N} \prod_{r=1}^R \Biggl(\tilde{\sigma}_r+v\sum_{k=1}^{K}\tilde{\theta}_{kr}u_k\Biggr)^{N_r} d\bs u\,dv
\end{equation}
with \,$\tilde{\theta}_{kr}=\alpha^{-1}_r\theta_{kr}$, and $\tilde{\sigma}_r=\sum_{k=K+1}^M \tilde{\theta}_{kr}$. From now on, and without loss of generality, we focus on $I_{\epsilon}(N)$ and to simplify notation use ${\theta}_{kr}$ and ${\sigma}_{r}$ in place of $\tilde{\theta}_{kr}$ and $\tilde{\sigma}_r$, subject to ${\theta}_{kr}\leq 1$ and ${\sigma}_{r}\leq 1$. Moreover, we assume that the ratios $\beta_r=N_r/N$ remain constant while increasing $N$\fix{, and that every class places a nonzero demand at some single-server node, $\sum_{k=1}^K\theta_{kr}>0$ for all $r$, so that the quantity $S_r(\bs u)=\sum_{k=1}^K\theta_{kr}u_k$ is not identically zero on $\Delta_K$}.

\subsection{Laplace's method}
\label{sec:laplaceapprox}
We first obtain the asymptotic approximation for $I_\epsilon(N)$ at $\hat{\bs u}$ in a model with single-server nodes only. 
This method requires the usual local regularity and global localization properties of Laplace's method~\cite{Kass90}. \fix{For the integral \eqref{eq:Iexp} at hand, we verify these properties directly, without attempting to establish the more general result in \cite{Kass90}.} After showing that
\begin{equation}
\label{eq:Iexp}
I_{\epsilon}(N)=  \int_{\mathbb{R}^{K-1}} e^{-N h_{N}(\bs x)} d\bs x
\end{equation}
for smooth and infinitely differentiable $h_{N}(\bs x)$, having constant order with respect to $N$ and bounded derivatives, we show that, for fixed $\epsilon>0$:
\squishlist
\item Condition 1: $I_{\epsilon}(N)$ exists and it is finite;
\item Condition 2: $h_{N}(\bs x)$ has a unique stationary point in the interior of the integration domain of \eqref{eq:Iexp}, which is its global minimum;
\item Condition 3: the Hessian of $h_{N}(\bs x)$ \fix{is positive definite} at its stationary point.
\squishend
\fix{In addition, the strict concavity and boundary behaviour established in the proof imply the standard localization property: outside any fixed neighbourhood of the stationary point the exponent is uniformly smaller, for all sufficiently large $N$.} Hence Laplace's method~\cite{Kass90} provides an asymptotic approximation \fix{with $\mathcal{O}(N^{-1})$ relative error}. Higher-order expansions may also be considered, but their computational cost grows quickly with the number of nodes in the model~\cite{Kass90}.
\begin{theorem}
\label{thm:laplace}
In a closed network without infinite servers, \fix{for every fixed \(\epsilon>0\), as \(N\to\infty\) while keeping \(N_r/N=\beta_r\) fixed,}
\begin{equation}
\label{eq:laplace}
I_{\epsilon}(N) = \sqrt{\frac{(2\pi)^{K-1}}{\det(\bs A)}}\prod_{r=1}^R \left(\sum_{k=1}^K \theta_{kr} \hat{u}_k\right)^{N_r}\prod_{i=1}^K \hat{u}^{1+\epsilon N}_i 
\end{equation}
where $\hat{\bs u}$ is the unique solution in $\Delta_K$ of the system of nonlinear equations
\begin{equation}
\label{eq:nonlinuv}
{
  \begin{aligned}
u_i&={\eta^{-1}_\epsilon(N)}\left(1+\epsilon N+\sum_{r=1}^R {\xi_r({\bs u}) {\theta}_{ir}u_i}\right)\quad i=1,\ldots,{K-1}
  \end{aligned}
  }
\end{equation}
with $\xi_r({\bs u})={N_r}(\sum_{k=1}^{K} {\theta}_{kr}u_k)^{-1}$, and where \fix{${\bs A}=[A_{ij}]_{i,j=1}^{K-1}$ is positive definite}, with entries
\begin{equation}
\label{eq:A}
A_{ij}=\begin{cases}
-\displaystyle\sum_{\substack{h=1\\ h\neq i}}^{K}
\left(\displaystyle\sum_{r=1}^R
\frac{N_r\theta_{ir}\theta_{hr}}
{\left(\sum_{k=1}^K\theta_{kr}\hat{u}_k\right)^2}
-\eta_\epsilon(N)\right)\hat{u}_i\hat{u}_h, & i=j,\\[3mm]
\displaystyle\left(\sum_{r=1}^R
\frac{N_r\theta_{ir}\theta_{jr}}
{\left(\sum_{k=1}^K\theta_{kr}\hat{u}_k\right)^2}
-\eta_\epsilon(N)\right)\hat{u}_i\hat{u}_j, & i\neq j.
\end{cases}
\end{equation}
for $i,j=1,\ldots,K-1$.
\end{theorem}
The last result provides by \eqref{eq:Geps} an approximation for $G^{\epsilon}_\t(\bs N)$. The role of the $\epsilon$ parameter is to ensure that the stationary point of the integrand of \eqref{eq:Gsimplex} remains in the interior of the integration domain \fix{uniformly in $N$ for fixed $\epsilon>0$}. For models with a finite $N$, the first condition holds irrespective of the value of $\epsilon$, which is needed only asymptotically, and choosing smaller values of $\epsilon$ generally returns more accurate results. \fix{We also show that, for any $\epsilon>0$, matrix $\bs A$ is positive definite, which is later used to develop a Monte Carlo integration method.}



Lastly, we note that  \eqref{eq:laplace} is a product. This is highly beneficial in applications, since we can avoid numerical difficulties associated to the rapid growth of the normalizing constant by directly computing $\log \G$. As a result, across thousands of models that we have solved in the numerical validation, we have never experienced numerical issues with \eqref{eq:laplace}. On the contrary, normalizing constant methods based on summations such as CA or \eqref{eq:grnmol} eventually fail on large models due to floating-point range exceptions and round-off errors.


\subsection{Proof of Theorem \ref{thm:laplace}}
\label{sec:proof}
The result follows by proving the assumptions of Laplace's method. Since $\sigma_r=0$, \eqref{eq:Idef} here simplifies to
\begin{equation}
\label{eq:IdefnoZ}
I_{\epsilon}(N)= \int_{\Delta_{K}} \prod_{i=1}^K u_i^{\epsilon N} \prod_{r=1}^R \Biggl(\sum_{k=1}^{K}{\theta}_{kr}u_k\Biggr)^{N_r} d\bs u
\end{equation}
We first apply an additive logistic transformation~\cite{Ait82}
\begin{equation}
\label{eq:logitu}
u_i = \begin{cases}
{e^{x_i}}{\left(1+\sum_{k=1}^{K-1} e^{x_k}\right)^{-1}} & i=1,\ldots,K-1\\
{\left(1+\sum_{k=1}^{K-1} e^{x_k}\right)^{-1}} & i=K
\end{cases}
\end{equation}
with Jacobian $$\left|\frac{\partial \bs u}{\partial \bs x}\right| = {\prod_{i=1}^{K-1} e^{x_i}}{\left(1+\sum_{k=1}^{K-1} e^{x_k}\right)^{-K}}$$
to obtain \eqref{eq:Iexp} with
\begin{equation*}
 - Nh_{N}(\bs x)=\sum_{i=1}^{K-1} (1+\epsilon N)x_i
+\sum_{r=1}^R N_r \log \left({\theta}_{Kr}+\sum_{k=1}^{K-1} {\theta}_{kr}e^{x_k}\right)
 -\eta_\epsilon(N)\log\left(1+\sum_{k=1}^{K-1} e^{x_k}\right)
\end{equation*}
Note that $h_{N}(\bs x)$ is smooth and infinitely differentiable. Moreover, $h_{N}(\bs x)$ has a constant order with respect to $N$ and its partial derivatives are also smooth and bounded at all orders. We are now ready to verify the conditions for Laplace's method given in Section~\ref{sec:laplaceapprox}.

\subsubsection{Condition 1: existence and finiteness}{
Since the logistic transformation does not affect existence and finiteness, it is sufficient to verify these properties on  \eqref{eq:IdefnoZ}. Existence follows since $\Delta_K$ is a finite domain and the integrand of \eqref{eq:IdefnoZ} exists at all points of $\Delta_K$. $I_\epsilon(N)$ is also finite for all $\epsilon>0$ and $N>0$, since we assumed throughout that ${\theta}_{kr}\leq 1$ and the domain $\Delta_K$ has a constant volume irrespective of the value of $N$.
}

\subsubsection{Condition 2: unique stationary point}
Condition 2 is verified as follows. We seek to solve $\nabla h_N(\bs x)=\bs 0$ and use the inverse transformation of \eqref{eq:logitu} to express the result over $\Delta_K$. The inverse transformation is given by \cite{Ait82}
\begin{equation}
\label{eq:invlog}
\hat{\bs x}=(\log \hat{u}_1/\hat{u}_K,\ldots,\log \hat{u}_{K-1}/\hat{u}_K)
\end{equation}
and yields the system of nonlinear equations \eqref{eq:nonlinuv}. We now show that this system admits a unique solution $\hat{\bs u}$. Moreover we also show that if $\epsilon>0$ then $\bs{\hat u}\in\Delta^{\epsilon}_{K,N} \subset \Delta_{K}$ for all $N>0$, with $$\Delta^{\epsilon}_{K,N}=\left\{\bs u| \bs u\in \Delta_K, u_i\geq (1+\epsilon N){\eta^{-1}_\epsilon(N)}, \forall i\right\}$$ Mapping back $\bs{\hat u}$ to $\mathbb{R}^{K-1}$ using the logistic transformation \eqref{eq:logitu}, this implies that the stationary point of $h_N(\bs x)$ is in the interior of the integration domain, {\em i.e.}, $\hat{\bs x}$ is finite in $\mathbb{R}^{K-1}$.

\fix{Since the Jacobian of \eqref{eq:logitu} equals $\prod_{i=1}^{K}u_i$, the integrand of \eqref{eq:Iexp} may be written in the original integration variables as $e^{-Nh_N(\bs x)}=e^{N\psi_N(\bs u)}$, with
\begin{equation}
\label{eq:psiN}
\psi_N(\bs u)=\left(\epsilon+\frac{1}{N}\right)\sum_{i=1}^K \log u_i+\sum_{r=1}^R \beta_r \log\left(\sum_{k=1}^K \theta_{kr}u_k\right)
\end{equation}
that is, $-Nh_{N}(\bs x)=N\psi_N(\bs u(\bs x))$ identically on $\mathbb{R}^{K-1}$. It is possible to verify that $\psi_N$ is strictly concave over $\Delta_K$ and that it diverges to $-\infty$ on the boundary of $\Delta_K$, thus it has its maximum at a unique point of the interior of $\Delta_K$, which is also its unique stationary point. Since \eqref{eq:logitu} places $\mathbb{R}^{K-1}$ in one-to-one correspondence with the relative interior of $\Delta_K$, it follows that $h_N$ has a unique stationary point $\hat{\bs x}$, which is finite in $\mathbb{R}^{K-1}$ and is the global minimum of $h_N$. It is further possible to show that the solutions of \eqref{eq:nonlinuv} in $\Delta_K$ coincide with the stationary points of $\psi_N$, so that $\hat{\bs u}=\bs u(\hat{\bs x})$ is the unique solution of \eqref{eq:nonlinuv}. Since each equation in \eqref{eq:nonlinuv} implies $u_i\geq (1+\epsilon N){\eta^{-1}_\epsilon(N)}$, and the same bound can be shown to extend to the $K$th coordinate, we conclude that $\hat{\bs u}\in\Delta^{\epsilon}_{K,N}$. In practice, $\hat{\bs u}$ may be computed with the fixed-point iteration}
\begin{equation}
\label{eq:nonlinuvINF}
u^{(n+1)}_i={\eta^{-1}_\epsilon(N)}\left(1+\epsilon N+\sum_{r=1}^R {{\xi}_r(\bs u^{(n)}) {\theta}_{ir}u^{(n)}_i}\right)
\end{equation}
for $i=1,\ldots,K-1$\fix{, $n\in \mathbb{N}$}.

\subsubsection{Condition 3: positive \fix{definite} Hessian}
\label{sec:proof3}
Let $\bs A=[A_{ij}]$, $i,j=1,\ldots,K-1$, be the Hessian of $Nh_N(\bs x)$ evaluated at the stationary point $\hat{\bs x}$. Computing $\bs A$ by the definition we obtain \eqref{eq:A}.
\fix{We now prove that $\bs A$ is positive definite for all $\epsilon>0$ and all $N>0$. Differentiating twice the identity $Nh_N(\bs x)=-N\psi_N(\bs u(\bs x))$, established in \eqref{eq:psiN}, and evaluating at $\hat{\bs x}$, it is possible to verify that the terms involving the second derivatives of the transformation \eqref{eq:logitu} cancel out, since $\sum_{k=1}^K u_k(\bs x)=1$ and the gradient of $\psi_N$ at a stationary point is orthogonal to $\Delta_K$. This leaves $\bs A=\bs J^T (-N\nabla^2\psi_N(\hat{\bs u}))\, \bs J$, in which $\bs J$ is the Jacobian matrix of \eqref{eq:logitu} at $\hat{\bs x}$, which has full rank. It is also possible to verify that $\nabla^2\psi_N(\hat{\bs u})$ is negative definite for all $\epsilon>0$ and $N>0$. Hence $\bs A$ is positive definite, which verifies Condition 3 and implies in particular $\det(\bs A)>0$.}

\subsubsection{Final expression}
{
\fix{For fixed $\epsilon>0$, the bound $\hat{u}_i\geq(1+\epsilon N)\eta_\epsilon(N)^{-1}$ keeps the stationary points in a compact subset of the interior of $\Delta_K$. Let $\psi_\infty$ be obtained from \eqref{eq:psiN} by replacing $\epsilon+N^{-1}$ with $\epsilon$. Then, growing $N$, we have $\psi_N\to\psi_\infty$ uniformly on compact subsets of the interior, while the term $\epsilon\sum_i\log u_i$ makes $\psi_N$ uniformly small near the boundary. Since $\psi_\infty$ is strictly concave with a unique maximizer, outside any fixed neighbourhood of $\hat{\bs x}$ the exponent is therefore separated from its maximum by a positive amount, uniformly for all sufficiently large $N$. Thus the integral concentrates at the unique non-degenerate stationary point.}
To conclude the proof of Theorem \ref{thm:laplace} we can apply Laplace's method in $\mathbb{R}^{K-1}$ to obtain the expansion
\[
I_{\epsilon}(N)=\sqrt{\frac{(2\pi)^{K-1}}{\det(\bs A)}} e^{-N h_N(\hat{\bs x})}\fix{\bigl(1+\mathcal{O}(N^{-1})\bigr)}
\]
The final expression \eqref{eq:laplace} follows by the expression of $h_N(\hat{\bs x})$ and the inverse transformation \eqref{eq:invlog}. By the initial definitions and using the $\alpha_r$ terms to remove the condition $\theta_{kr}\leq 1$, the asymptotic expansion is finally given by
\begin{equation}
G^{\epsilon}_\t(\bs N)=\frac{\Gamma(\eta_\epsilon(N))}{N_1!\cdots N_R!(\Gamma(1+\epsilon N))^K}\,{\sqrt{\frac{(2\pi)^{K-1}}{\det(\bs A)}}\prod_{r=1}^R \left(\sum_{k=1}^K \theta_{kr} \hat{u}_k\right)^{N_r}\prod_{i=1}^K \hat{u}^{1+\epsilon N}_i}\fix{\bigl(1+\mathcal{O}(N^{-1})\bigr)}
\end{equation}
}

\subsection{Extensions}

\subsubsection{Models with infinite servers}
\label{sec:Z}
When the population at one or more nodes does not scale asymptotically, the asymptotic stationary point of $h_N(\bs x)$ does not lie anymore in the interior of the integration domain and this complicates the asymptotic validity of Laplace's method. In the presence of infinite servers, it does not seem easy to scale parameters in \eqref{eq:Iexp} to address the problem. 
Thus, in models with infinite servers we propose Laplace's method only as a sub-asymptotic heuristic. The sub-asymptotic method amounts to fitting the integrand of \eqref{eq:Iexp} to a multivariate normal density and using the resulting closed-form expressions to approximate $I_{\epsilon}(N)$, for $N<\infty$. This needs to be coupled with an approximation of the indefinite integral in \eqref{eq:GZsimplex}.

The heuristic follows a very similar argument as in the case without infinite servers, thus we just give a sketch. We first apply the change of variable $v=e^{x_0}$ and the logistic transformation \eqref{eq:logitu} so that
\[
I_{\epsilon}(N)=\int_{\mathbb{R}^K} b_N e^{-Nh_N(\bs x)}d\bs x
\]
with $\bs x=(x_1,\ldots,x_{K-1},x_0)$, $b_N=(\Gamma(\eta_\epsilon(N)))^{-1}$, and
\begin{multline*}
-Nh_N(\bs x)=-e^{x_0}+K(1+\epsilon N)x_0+\sum_{r=1}^R N_r \log\left(\sigma_r+e^{x_0}\theta_{Kr}+\sum_{k=1}^{K-1} e^{x_k} (\sigma_r+\theta_{kr}e^{x_0})\right)
\\+(1+\epsilon N)\sum_{k=1}^{K-1}{x_k} - \eta_\epsilon(N)\log \left(1+\sum_{k=1}^{K-1} e^{x_k}\right)
\end{multline*}
Note that this is a $K$-dimensional integral, whereas in the case without infinite servers we have used $K-1$ dimensions. Setting $\nabla h_N(\bs x)=0$ and simplifying terms, we find that the stationary point is written in terms of the original integration variables as the solution $(\hat{\bs u},\hat{v})$ of the system
\begin{equation}
\begin{aligned}
\label{eq:vfpi}
u_i&={\eta^{-1}_\epsilon(N)}\left(1+\epsilon N+\sum_{r=1}^R {\xi_r({\bs u},v) (\sigma_r+v{\theta}_{ir})u_i}\right)\qquad \forall i\neq K\\
  v&=\eta_\epsilon(N)\fix{-\sum_{r=1}^R \xi_{r}(\bs u,v) \sigma_r}
\end{aligned}
\end{equation}
where $\bs u\in \Delta_K$, $\eta_\epsilon(N)=N+K(1+\epsilon N)$, $\xi_{r}(\bs u,v)=N_r (\sigma_r+v\sum_{k=1}^K \theta_{kr}u_k)^{-1}$ and in which the equation for $v$ uses that $\sum_{r=1}^R \xi_{r}(\bs u,v) (\sigma_r+v\sum_{k=1}^K \theta_{kr}u_k)=N$.

Explicit formulas for the entries of $\bs A$ are obtained by computing the Hessian matrix of $Nh_N(\bs x)$ expressed in terms of the variables $(\hat{\bs u},\hat{v})$. Define $\hat{\theta}_{kr}=\sigma_r+\hat{v}\theta_{\fix{kr}}$, $\forall k,r$, the entries of $\bs A=[A_{ij}]$ are given by 
\allowdisplaybreaks[1]
\begin{align*}
A_{ii}=&-\fix{\sum_{j\neq i} A_{ij}}\\
A_{ij}=&\fix{\left(\sum_{r=1}^R \frac{\xi^2_{r}(\hat{\bs u},\hat{v})}{N_r}{\hat{\theta}_{{ir}}\hat{\theta}_{jr}} -\eta_\epsilon(N)\right)}{\hat{u}_i\hat{u}_j}  \qquad i\neq j\\
A_{00}=&\hat{v}\left(1-\sum_{r=1}^R \frac{\xi^2_{r}(\hat{\bs u},\hat{v})}{N_r}\sigma_r\left(\sum_{k=1}^K \theta_{kr}\hat{u}_k\right)\right)\\
A_{i0}=&\,A_{0i}=\hat{v}\hat{u}_i \sum_{r=1}^R \left(\frac{\xi^2_{r}(\hat{\bs u},\hat{v})}{N_r}{\hat{\theta}_{ir}\left(\sum_{k=1}^K \theta_{kr}\hat{u}_k\right)}-{\xi_{r}(\hat{\bs u},\hat{v})} \theta_{ir}\right)
\end{align*}
for all $i,j=1,\ldots,K-1$. The above expressions use the inverse transformations \eqref{eq:invlog} and $x_0=\log v$.
The knowledge of $(\hat{\bs u},\hat{v})$ and $\bs A$ provides a Laplace-type approximation for \eqref{eq:GZsimplex}
\begin{equation}
\label{eq:laplaceZ}
I_{\epsilon}(N) \approx \frac{e^{-\hat{v}}\hat{v}^{K(1+\epsilon N)}}{\Gamma(\eta_\epsilon(N))} \sqrt{\frac{(2\pi)^{K}}{\det(\bs A)}} \prod_{r=1}^R \left(\sigma_r + \hat{v}\sum_{k=1}^K \theta_{kr} \hat{u}_k\right)^{N_r}\prod_{k=1}^K \hat{u}^{1+\epsilon N}_k
\end{equation}
where the exponent of $\hat{v}$ includes the contribution of the Jacobian. Equation \eqref{eq:laplaceZ} can be readily applied to approximating models with infinite server nodes, leading to
\begin{equation}
G^{\epsilon}_\t(\bs N)=\frac{e^{-\hat{v}}\hat{v}^{K(1+\epsilon N)}}{N_1!\cdots N_R!(\Gamma(1+\epsilon N))^K}\sqrt{\frac{(2\pi)^{K}}{\det(\bs A)}} \prod_{r=1}^R \left(\sigma_r + \hat{v}\sum_{k=1}^K \theta_{kr} \hat{u}_k\right)^{N_r}\prod_{k=1}^K \hat{u}^{1+\epsilon N}_k
\end{equation}
As before, the $\epsilon>0$ parameter should be chosen as small as possible, but such that $\det(\bs A)>0$.

\subsubsection{Combining Laplace's method with AMVA}
In some inference problems, measurements for both system state and mean performance metrics are available. In this case, in addition to likelihood-based inference, one may require that the estimated model also matches the empirical mean value of some performance metrics. This typically requires to run an AMVA algorithm alongside the likelihood maximization algorithm. 

In this section we argue that a more efficient way to optimize these models is to heuristically compute ${\hat{\bs u}}$ using the results of AMVA. This effectively doubles the speed of the optimization, since the same AMVA fixed-point iteration can be used both to calculate likelihood and mean measures. A limitation of this method is that no formal guarantee is in place to ensure that $\det(\bs A)>0$ on all instances. However, no problematic instance in this sense is observed throughout the numerical validation. In cases where $\det(\bs A)\leq 0$, one may try to heuristically increase $\epsilon$ in order to resolve this issue.

The proposed method works as follows. Let us observe that, as $N\to\infty$, \eqref{eq:nonlinuv} tends to
\begin{equation}
{
  \begin{aligned}
u_i&=\frac{\epsilon}{1+K\epsilon}+\sum_{r=1}^R \hat{\xi}_r(\hat{\bs u}) {\theta}_{ir}u_i \quad i=1,\ldots,K
  \end{aligned}
  }
\end{equation}
where $\hat{\xi}_r({\bs u})={\beta_r(1+K\epsilon)^{-1}}(\sum_{k=1}^{K} {\theta}_{kr}u_k)^{-1}$.
As $\epsilon\to 0$ the last expression coincides with the expression of the mean-value analysis algorithm's queue-length equations when ${u}_i$ is the total queue-length at node $i$ divided by $N$, and $N\to \infty$. This suggests the following way to determine the point ${\hat{\bs u}}$ at which we instantiate the Laplace's method. Instead of using \eqref{eq:nonlinuv}, we choose $\hat{\bs u}\approx \hat{\bs q}=(\hat{q}_{1},\ldots,\hat{q}_M)$, in which $\hat{q}_{i}=\sum_r q_{ir}(\bs N)/N$, $\forall i$, where  $q_{ir}(\bs N)$ is the mean queue-length of class $r$ at node $i$ in a model with population $\bs N$, a value which can be accurately approximated in $\mathcal{O}(1)$ time and space using AMVA~\cite{Sch79}.


\subsubsection{Monte Carlo integration}
\label{sec:mci}
The applicability of Laplace's method indicates that the integral $I_{\epsilon}(N)$ may be approximated using a multivariate normal distribution centered at the stationary point of $h_N(\bs x)$ and with covariance matrix $\bs A^{-1}$. The resulting normal distribution is non-degenerate if $\bs A$ is positive definite. In situations where asymptotic expansions are expensive or inaccurate, one may thus apply an importance sampling method to $I_{\epsilon}(N)$, using samples from a multivariate normal distribution. Denote by $\bs x_j$ the $j$th sample drawn, out of a total of $J$. We have the importance sampling estimator
\begin{equation}
\label{eq:lse}
I_{\epsilon}(N) \approx J^{-1}\sum_{j=1}^J \frac{b_N e^{-Nh_N(\bs x_j)}}{\phi(\bs x_j)} \qquad \bs x_j \sim \mathcal{N}(\hat{\bs x},\bs A^{-1})
\end{equation}
where $\hat{\bs x}$ is the stationary point of $h_N(\bs x)$ and $\phi(\cdot)$ stands for the normal density function. For the case without infinite servers, $b_N=1$ and $h_N(\bs x)$ is defined as in Section~\ref{sec:proof}. For models with infinite servers, one needs to use the expressions of $b_N$ and $h_N(\bs x)$ given in Section~\ref{sec:Z}. From Monte Carlo integration theory, \eqref{eq:lse} converges to $I_{\epsilon}(N)$ as $O(J^{-1/2})$ under a growing sample size $J$.

\section{Numerical results}
\label{sec:res}

\subsection{Algorithms}
In this section we assess accuracy and speed of the proposed methods. We distinguish the proposed algorithms in two groups:
\squishlist
\item {\em Deterministic methods}, such as the asymptotic expansions, which return the same answer in successive invocations on the same model and therefore are suitable for use within deterministic optimization programs. We include in this group RAY and the asymptotic expansion \eqref{eq:Iexp}, referred to as the {\em logistic expansion} (LE), and the heuristic variant of LE calibrated with AMVA, denoted by LE-A. We also consider in this group the cubature rules given in \eqref{eq:grnmol} with $S=\{1,3,5,7\}$ and denote, {\em e.g.}, by CUB5 a cubature rule with $S=5$. We have also experimented with TE, CA, and MoM, but computational times are far larger than those of the other methods and incompatible with the scale of the experimental validation, which encompasses thousands of optimization programs.
\item {\em Randomized methods}, which use sampling to achieve the desired accuracy in return for an increased effort. We include in this group MCI and the Monte Carlo integration method in \eqref{eq:lse}, which we call {\em logistic sampling} (LS). Monte Carlo methods are instantiated with $J\in\{10^1,10^2,10^3\}$ samples, {\em e.g.}, MCI2 stands for MCI with $J=10^2$, and similarly LS3 has $J=10^3$.
\squishend
For deterministic methods, we are interested in assessing both accuracy and ability to guide optimization-based search. For randomized methods, we verify accuracy as the number of samples grows. Remarks on the implementations are as follows:
\squishlist
\item In LE we use fixed-point iteration to solve \eqref{eq:nonlinuvINF}, setting the convergence tolerance on the 1-norm of $\hat{\bs u}$ to $\tau=10^{-10}$. The initial point has $u_i=1/K$, $\forall i$. We also set $\epsilon=10^{-10}$. Out of the thousands of models solved, none failed to converge and none required to increase $\epsilon$ beyond its initial value.
\item LE-A is implemented using the Bard-Schweitzer AMVA~\cite{Sch79} for determining the point $\hat{\bs q}$ with a $\delta=10^{-6}$ convergence tolerance on the 1-norm of the mean queue-lengths. Also in this case none of the models failed to converge.
\item For cases where $\bs A$ is not positive definite, LS is instantiated by increasing $\epsilon$ in steps of $10^{-3}N$, until obtaining a positive definite matrix. In small and medium-sized models, this calibration is not normally required. However, on large models where the entries of $\bs A$ are small, very few increments of $\epsilon$ are normally sufficient to address the issue. In the random validation on large models, this calibration is required on $44\%$ of the instances and occurs prior to computing \eqref{eq:lse}. The computational cost of the calibration is negligible.
\squishend

\subsection{Methodology}
An important issue for the validation methodology is that for large models $\G$ cannot be computed exactly due to the large cost of the exact algorithms. Thus we first carry out a validation on small and medium-sized models where the normalizing constant can be obtained exactly. Afterwards, we report a similar validation on larger models where we estimate $\G$ by Monte Carlo integration with a large number of samples ($10^7$). In the large-scale setting, we attempt to compensate the variance of the estimator of the normalizing constant by assessing percentage error with respect to the scale, {\em i.e.}, $\log \G$. Note that on most large-scale models the order of the normalizing constant is the dominant factor in the likelihood expressions.

The validation does not include mean performance metrics. This is because their computation can be performed very efficiently using AMVA methods~\cite{Sch79}. Our methods are instead proposed for the accurate computation of likelihoods and probabilities, which require $\G$ and are still difficult to compute in practice.

{ The computational times to run LE and LE-A inside optimization programs are very small, typically a fraction of a second. This is due to the rapid convergence of the fixed-point algorithms used to determine the location of the stationary point. For example, on the largest model with $K=R=64$ and $N=4096$ LE takes $26s$ to find at the first iteration the stationary point, but just $2s$ for a new prediction after a $10^{-6}$ increment of $\bs \theta$, provided that the fixed point equations~\eqref{eq:nonlinuvINF} are re-initialized at the previously-found stationary point. When the model size is decreased to $K=R=8$ the first solution requires just $0.21s$, while successive updates about $0.015s$. Since the optimization-based study considers an identical timeout for all methods, we do not provide details on the running times of individual algorithms. Lastly, we remark that space complexity is negligible and does not grow significantly with the model size. This is because all approximation methods considered throughout have $\mathcal{O}(1)$ space complexity as the population sizes $N$ grow, with $N$ typically being the largest model parameter.}

\subsection{Computing a single normalizing constant}
\subsubsection{Small and medium-sized models}
\label{sec:small}
We consider randomly-generated models with $K\in\{2,4,6\}$ nodes, $R\in\{2,4,6\}$ classes, and where each class has the same number of jobs equal to $N/R=\{2,4,8\}$. Thus the largest model in this group has 6 nodes, 6 classes and 48 jobs. We use less jobs than in the motivating example in Table~\ref{tab:motivating} since we now consider models with $R=6$ classes that are much more expensive to solve exactly. For any given triplet $(K,R,N)$, we solve $100$ random instances, for a total of 2700 models. In each instance, demands are generated at random in $[0,1]$. $\G$ is computed exactly using the CA algorithm.

Note that RAY is the only method that incurs failures during execution. This occurs on 9 models out of 2700 and it is due to a singular determinant in its expression. Indeed, RAY does not provide correctness guarantees in the sub-asymptotic setting~\cite{KneT92}. We count as a failure a run that either stops due to excessive memory requirements, or that returns a 0, {\small \sf NaN}, $\pm \infty$, or a complex value for $\G$ due to numerical issues.

Tables \ref{tab:res-small} and \ref{tab:res-small-rand} give the {\em mean absolute percentage error} (MAPE) on $\G$ for deterministic and randomized methods. The results indicate that the CUB dominates all other methods, including the randomized ones. Execution times of CUB on these models are in the order of a few milliseconds, making this method preferable in small and medium-sized models. As expected, asymptotic expansions incur smaller errors as the number of jobs grows. The Monte Carlo integration methods, MCI and LS, are instead more accurate with fewer jobs. However, increasing the number of samples in both methods quickly lowers errors to the desired level.

In order to better understand the differences between MCI and LS, we have investigated how the accuracy of the two methods varies under increasing number of nodes $K$ or number of classes $R$. MCI decreases errors as $K$ increases, whereas it performs worse under increasing $R$. For example, MCI1 goes from 35.4\% to 48.2\% as the number of classes goes from 2 to 6, whereas it decreases errors from 57.4\% to 30.1\% when the number of nodes grows of the same amount. Conversely, LS is rather insensitive to $R$, with LS1 error lying between 22.5\% and 26.4\%, but the method incurs larger errors as $K$ grows, with LS1 going from 12.5\% with 2 nodes to 38.9\% with 6 nodes. This increased error is due to the larger number of integration dimensions in \eqref{eq:Iexp}, which requires a larger $N$ value to deliver a similar level of accuracy. Similar trends are seen also in large models and suggest that the two methods can complement each other, preferring MCI on models with many nodes and LS on models with several classes.

{\begin{table}[t]
\caption{{Small and medium models - deterministic methods}}
{\begin{center}
\begin{tabular}{lcccc}
\toprule
$K,R\in\{2,4,6\}$ & \multicolumn{3}{c}{{ MAPE (\%)}}\\
\midrule
 $N/R=$ & {\em 2} & {\em 4} & {\em 8} \\
\midrule
 CUB1 & 0.0 & 0.4  & 8.9  \\
 CUB3 & 0.0 & 0.4  & 7.4  \\
 CUB5 & 0.0 & 0.0  & 0.7  \\
 CUB7 & 0.0 & 0.0  & 0.1   \\
 LE & 25.7 & 25.2  & 24.1 \\
 LE-A & 26.7 & 23.3  & 18.3  \\
 RAY & 362.0 & 142.7  & 96.5  \\
\bottomrule
\end{tabular}
\end{center}
}
\label{tab:res-small}
\end{table}
\begin{table}[t]
\caption{Small and medium models - randomized methods}
{\begin{center}
\begin{tabular}{lcccc}
\toprule
$K,R\in\{2,4,6\}$ & \multicolumn{3}{c}{{ MAPE (\%)}}\\
\midrule
 $N/R=$ & {\em 2} & {\em 4} & {\em 8} \\
\midrule
 LS1 & 23.1 & 21.1  & 31.6   \\
 LS2 & 12.1 & 12.9  & 12.1  \\
 LS3 & 6.0 & 7.0  & 7.4   \\
 MCI1 & 28.9 & 41.7  & 55.5   \\
 MCI2 & 9.3 & 13.0  & 18.1   \\
 MCI3 & 3.0 & 4.1 & 5.3   \\
\bottomrule
\end{tabular}
\end{center}
}
\label{tab:res-small-rand}
\end{table}
\begin{table}[t]
\caption{Small and medium models with infinite server nodes - deterministic methods}
{\begin{center}
\begin{tabular}{lccccc}
\toprule
$K,R\in\{2,4,6\}$ & \multicolumn{4}{c}{{ MAPE (\%)}}\\
\midrule
 $\sigma_r$ & $0.1\theta_{max}$ & $\theta_{max}$ & $10\theta_{max}$& $100\theta_{max}$ \\
\midrule
 PAN & {\small \sf N/A} & {\small \sf N/A} & 26.2 & 0.1 \\
 LE & 7.1 & 4.8 & 2.9 & 1.5 \\
\bottomrule
\end{tabular}
\end{center}
}
\label{tab:res-small-rand-Z}
\end{table}
}



\begin{table}[t]
\caption{Large models - deterministic methods}
{
\begin{center}
\begin{tabular}{lcccccc}
\toprule
$K,R\in\{16,32,64\}$ & \multicolumn{6}{c}{{ MAPE (\%)}} \\
\midrule
 $N/R=$ & {\em 2} & {\em 4} & {\em 8} & {\em 16} & {\em 32} & {\em 64} \\
\midrule
CUB1 & 0.0 & 0.1  & 0.4 & 1.1 & 2.5  & 2.9 \\
 LE & 1.8 & 1.0  & 0.7 & 0.4 & 0.4 & 0.5 \\
 LE-A & 1.8 & 1.0 & 0.7 & 0.4 & 0.4 & 0.4  \\
 RAY & 14.5 & 4.8 & 1.4 & 0.5 & 0.3 & 0.4 \\
\bottomrule
\end{tabular}
\end{center}
}
\label{tab:res-large}
\end{table}

\begin{table}[t]
\caption{Large models - randomized methods}
{
\begin{center}
\begin{tabular}{lcccccc}
\toprule
$K,R\in\{16,32,64\}$  & \multicolumn{6}{c}{{ MAPE (\%)}} \\
\midrule
 $N/R=$ & {\em 2} & {\em 4} & {\em 8} & {\em 16} & {\em 32} & {\em 64} \\
\midrule
 LS1 & 0.9 & 0.6  & 1.4 & 2.4 & 2.3 & 1.8  \\
 LS2 & 0.5 & 0.5  & 1.3 & 2.2 & 2.2 & 1.8   \\
 LS3 & 0.3 & 0.4  & 1.2 & 2.1 & 2.1 & 1.7  \\
 MCI1 & 0.1 & 0.1  &0.2 &0.3 & 1.5 &6.1  \\
 MCI2 & 0.0 & 0.0  &0.1 &0.1 & 0.4 &3.3  \\
 MCI3 & 0.0 & 0.0 & 0.0 &0.0 & 0.2 &1.6  \\
\bottomrule
\end{tabular}
\end{center}
}
\label{tab:res-large-rand}
\end{table}

\subsubsection{Large-scale models}
\label{sec:large}
We now consider a similar setup as in the previous experiment, but with $K$, $R$ ranging in $\{16,32,64\}$. The number of jobs ranges in $N/R=\{2,4,8,16,32,64\}$. Thus the largest models have $64$ nodes, $64$ classes and $4096$ jobs. As explained before, in this setting it is difficult to obtain the exact value of the normalizing constant, thus we compare against MCI with $10^7$ samples and focus on matching $\log \G$. The MAPE on $\log \G$ may be seen as the percentage error introduced in log-likelihoods such as \eqref{eq:llog}.

Tables \ref{tab:res-large} and \ref{tab:res-large-rand} present the results of these experiments. Since CUB3, CUB5 and CUB7 fail on over 98\% of the large instances due to floating-point range exceptions, the corresponding entries are omitted from the table. The results indicate that CUB1 remains the best method under small population sizes, however as $N$ grows the asymptotic expansions become the most accurate. The RAY method is the least accurate in light load, but has a similar accuracy to the other methods in high-load. This is indeed the regime that matches the assumptions for the scaling used in RAY \cite{KneT92}; we have noted however that on models with a large number of classes, but a few queues, RAY is less accurate than LE and LE-A, which is consistent with the fact that the scaling used in \cite{KneT92} assumes a growing number of nodes. For example, going from $K=16$ to $K=64$ nodes RAY improves its error from an average of $6.82\%$ to $1.28\%$. On the opposite, when the number of classes is increased from $R=16$ to $R=64$, RAY goes from an average error of $0.99\%$ to $7.22\%$. In the same ranges for nodes and classes, LE and LE-A have narrow error bands between $0.46\%$ and $1.23\%$ average error.

\subsubsection{Models with infinite server nodes}
We have repeated the experiments in Section~\ref{sec:small} on models with a infinite server node, focusing on the validation of the heuristic given in Section~\ref{sec:Z}. Since with infinite servers RAY is no longer applicable, we have validated LE against the PANACEA (PAN) asymptotic expansion~\cite{McKM84}. We have set in the experiments an identical think time on all classes equal to $\sigma_r\in\{0.1\theta_{max}$, $\theta_{max}$, $10\theta_{max}$, $100\theta_{max}\}$, where $\theta_{max}=\max_{r}\max_{k=1,\ldots,K}\theta_{kr}$. Results are shown in Table~\ref{tab:res-small-rand-Z}. PAN fails on all models with think time $\sigma_r=0.1\theta_{max}$ and $\sigma_r=\theta_{max}$ due to violation of the normal usage assumption; it instead returns a $26.2\%$ MAPE with $\sigma_r=\theta_{max}$, and an error less than $0.1\%$ with $\sigma_r=10\theta_{max}$. LE returns a valid solution in all cases, with decreasing errors for increasing $\sigma_r$ values. Thus, LE appears generally more robust than PAN, which is preferable only in very lightly loaded models. 

\subsection{Optimization programs}
We now compare the methods against the likelihood maximization problem \eqref{eq:llog} for service demand estimation, focusing on deterministic methods. For the sake of illustration of the limited performance of randomized methods in this setting, we also include MCI3 in the validation. We consider problems with $K,R\in \{2,4,8,16,32\}$ and populations with $N/R\in\{2,20,40\}$ jobs per class. Each experiment is carried out with the same procedure 
 described in Section~\ref{sec:running}, in particular setting a timeout of $T=10$ minutes. We repeat the same experiment 15 times randomizing demands, solving 1125 optimization programs for each method. Upon detecting an invalid normalizing constant the interior point method returns a failure.  We also mark as failed all the runs returning demands that do not satisfy the constraint $\t\geq \bs 0$.

\subsubsection{Metric}
As mentioned in Section~\ref{sec:running}, a critical issue in the analysis of the results is that \eqref{eq:llog} is non-convex, thus the choice of the initial points affects the relative error on the $\t$ estimate, irrespective of the quality of the approximation of $\G$. This is addressed in Section~\ref{sec:running} by comparing results against an exact method, run without timeout. Here we cannot apply the same approach due to the size of the models and also MCI with a large number of samples ($J=10^7$) fails since the variance of the estimator adversely affects the search direction of the interior-point method. Thus, on most models it does not seem possible to determine and compare the methods based on an absolute accuracy metric.

To cope with this problem, we use the same initial point for all the methods and compare them {relatively} to each other. For each model, we rank methods based on the absolute percentage error from the true value of $\t$. In this way, a method that achieves the best possible estimator given the initial point, will be ranked first, irrespective of the magnitude of the error that depends on the initial point and the local optimum found. Methods that return the same demands are assigned the same rank.

\begin{figure*}[t!]
\centering
\subfigure [$N/R\in\{2,20,40\}$] {
 {\resizebox{0.35\linewidth}{!}{\includegraphics[]{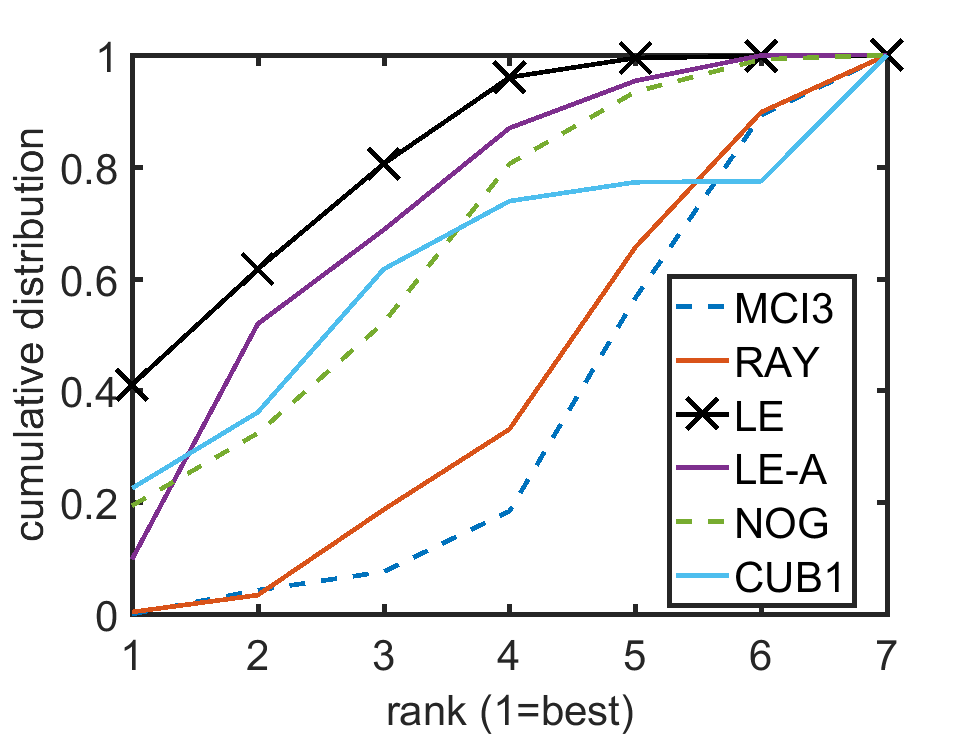}}}
 }
\subfigure [$N/R=2$] {
 {\resizebox{0.35\linewidth}{!}{\includegraphics[]{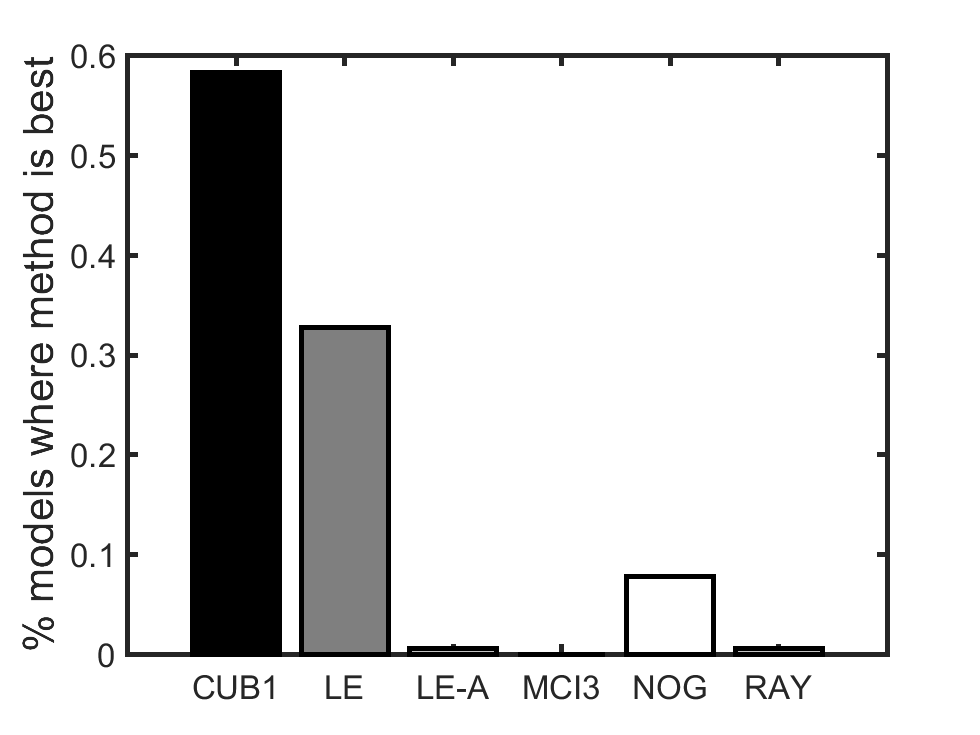}}}
 }\\
\subfigure [$N/R=20$] {
 {\resizebox{0.35\linewidth}{!}{\includegraphics[]{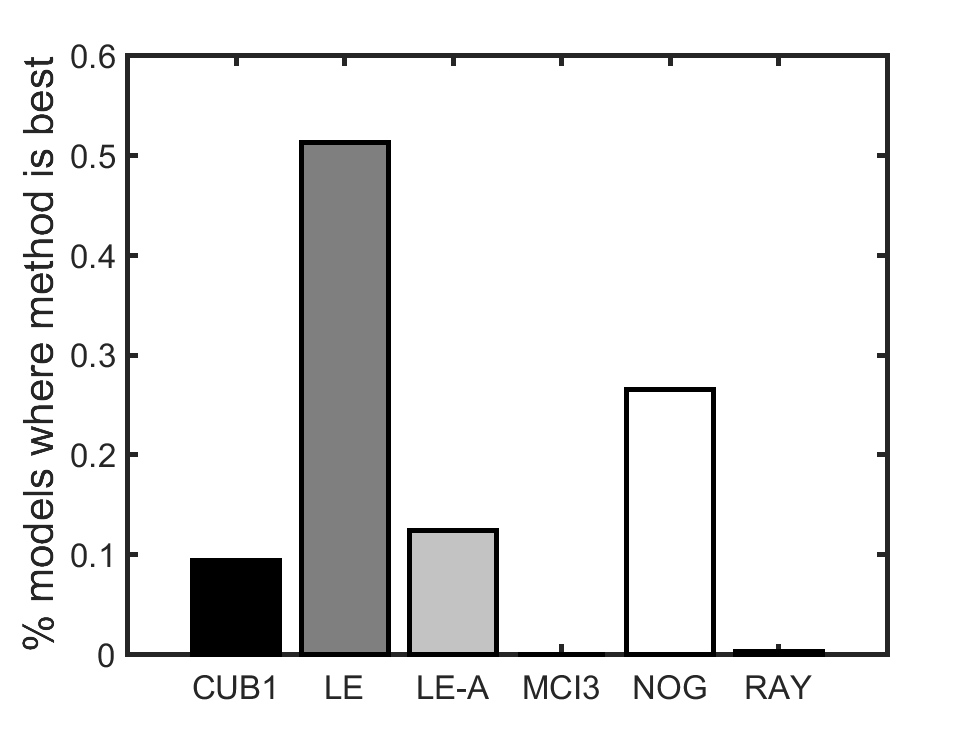}}}
 }
\subfigure [$N/R=40$] {
 {\resizebox{0.35\linewidth}{!}{\includegraphics[]{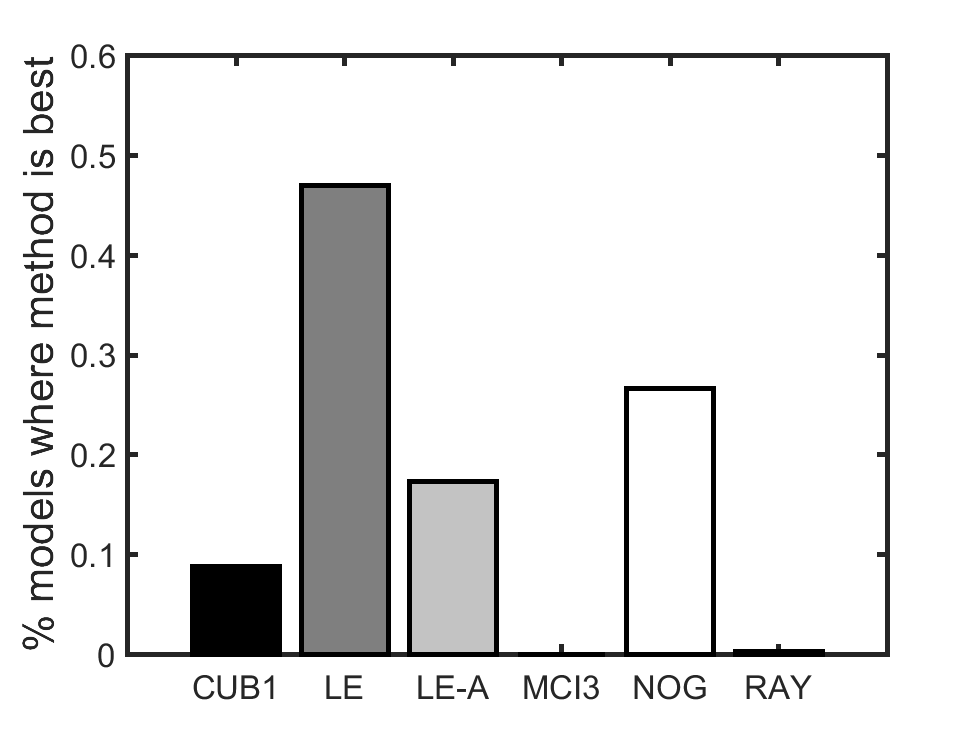}}}
 }
 \caption{Optimization programs: experimental results}
 \label{fig:case-study}
\end{figure*}

\subsubsection{Results}
Experimental results are given in Figure~\ref{fig:case-study}. We include in the study also the NOG method, which neglects the normalizing constant. In Figures \ref{fig:case-study}(b)-(d) we show how frequently each model is ranked best for a given population level $N$. The results indicate that LE outperforms all the other methods and it is slightly better than LE-A. However, in models with a small number of jobs CUB is preferable, which is consistent with the observations in Section~\ref{sec:small}. The fair performance of NOG is explained by the fact that asymptotically the closed network approaches an open network, where the arrival rate intensity matches the cumulative departure rate from the bottleneck nodes. The methods proposed in this paper remain preferable to NOG, as they are the best ones in most models. This is evident in Figure \ref{fig:case-study}(a), which indicates that LE is the best method among the considered ones. The figure shows that about 41\% of the times LE is the best method, and in about 62\% of the cases it ranks second, typically behind CUB1, NOG, or LE-A. Table~\ref{tab:fail} reports statistics on the number of failures, which occur only for CUB as the load grows and for RAY, similarly to what seen for small and medium-sized models. Methods not shown in the table do not incur failures.

\begin{table}
\caption{Percentage of failures (1125 models).}
\begin{center}
{
\begin{tabular}{lcccc}
\toprule
Method & Total & $N/R=2$ & $N/R=20$ & $N/R=40$ \\
\midrule
CUB1 & 22.4 & 0.0 & 26.9 & 40.3\\
RAY  & 8.8 & 2.7 & 16.0 & 7.7 \\
\bottomrule
\end{tabular}
}
\end{center}
\label{tab:fail}
\end{table}

\subsubsection{Single-class models}
\label{sec:single}
{ Lastly, for completeness we include results concerning inference in single-class models, which also arise in practice. For such models we do not study the computation of a single normalizing constant since exact $\mathcal{O}(1)$ expressions are available, {\em e.g.} \eqref{eq:exsingle}. We instead consider likelihood maximization with $K\in\{2,4,8,16,32\}$ nodes, $R=1$ class, $N=\{2,20,40\}$, and single-class demands $\theta_k=k/K$. The initial guess for the demand matrix is $\theta_k=1/K$, $\forall k$. We include in the study the exact-order asymptotic (EOA) formula recently proposed in \cite{GeoXS12}.

Generally speaking, exact methods such as CA are very fast on single-class models, and return a tiny error, on average just $0.058\%$. However, CA complexity is $\mathcal{O}(N)$, thus in models with very large populations $N$ approximations may be of interest. In the above study, LE returns a MAPE of just $0.28\%$. The other methods are instead rather inaccurate, with a MAPE of 28.1\% for CUB1, 26.0\% for RAY, 24.8\% for LE-A, and 66.1\% for EOA. This overall indicates that LE is fit for use also in single-class problems, although exact methods such as CA seem sufficient in practice.}

\subsection{Summary}
Summarizing, the results indicate the following main properties for the proposed algorithms:
\squishlist
\item On small models, CUB dominates all other algorithms.
\item On large models, one should choose CUB if the population is small, or otherwise prefer LE. The same criteria applies to optimization programs involving likelihoods.
\item Among randomized methods, LS is the best on models with many classes, whereas MCI is best with several nodes.
\squishend






\section{Conclusion}
\label{sec:conclusion}
This paper has shown that performance inference over closed systems faces computational hurdles. If the model is a closed product-form multiclass network, we have shown that computational issues can be addressed by novel asymptotic expansions and Monte Carlo sampling methods for the normalizing constant of state probabilities. 
Future research may further investigate the implications of the integral form \eqref{eq:Gsimplex} for the exact theory of normalizing constants, for example on models with multi-server and load-dependent nodes. 

\section*{Acknowledgement}
This research has been partially funded by the European Union's Horizon 2020 research and innovation programme under grant agreement No. 644869 (DICE) and by a UK Engineering and Physical Sciences Research Council grant (EP/L00738X/1). Research data is available at ({\small \url{https://doi.org/10.5281/zenodo.546873}}) under CC-BY 4.0 licence. The author wishes to thank Tony Field for support while preparing this work and Urtzi Ayesta for his helpful comments while serving as managing editor for this paper.

{\footnotesize \bibliographystyle{abbrv}

}
\appendix



\section{Divided differences}
\label{apx:L1}
We show that $g_{\t}(N)=[\theta_1,\mathellipsis,\theta_{K}]x^{N+K-1}$ holds for arbitrary demands $\theta_1\leq \theta_2\leq \cdots \leq \theta_K$. In the case of distinct demands, the result follows by \cite[Sec. 1.31]{Mil33}, \cite{GruM78}. Otherwise note that divided differences may be recursively defined as follows~\cite{DeB04}
\begin{equation}
\label{eq:recdefdd}
 [\theta_1,\ldots,\theta_K]f(x)=
\begin{cases}
\dfrac{f^{(K-1)}(\theta_1)}{(K-1)!} & \mbox{if }\theta_1=\theta_K\\
\dfrac{[\theta_2,\ldots,\theta_K]f(x)-[\theta_1,\ldots,\theta_{K-1}]f(x)}{\theta_K-\theta_1} & \mbox{otherwise}
\end{cases}
\end{equation}
where $f^{(n)}(x)$ denotes the $n$th derivative of $f(x)$. Let us then note the following relation~\cite{Har85}
\begin{align}
\label{eq:harrec}
g_{\t}(N)&=\frac{\theta_K g^{-1}_{\t}(N)-\theta_1 g^{-K}_{\t}(N)}{\theta_K-\theta_1 } \nonumber\\
&=\frac{g^{-1}_{\t}(N+1)-g^{-K}_{\t}(N+1)}{\theta_K-\theta_1 }
\end{align}
where $g^{-i}_{\t}$ refers to a network without node $i$ and we have used that by convolution $g^{-1}_{\t}(N+1) = g^{-1-K}_{\t}(N+1) + \theta_K g^{-1}_{\t}(N)$ and $g^{-K}_{\t}(N+1) = g^{-1-K}_{\t}(N+1) + \theta_1 g^{-K}_{\t}(N)$, see {\em e.g.}~\cite{Buz73}.

We can now show that \eqref{eq:recdefdd} and \eqref{eq:harrec} are equivalent if $f(x)=x^{N+K-1}$. In the case $\theta_1\neq \theta_K$, \eqref{eq:recdefdd} readily matches \eqref{eq:harrec}. Since $\theta_1\leq \theta_2\leq \cdots \leq \theta_K$, upon reaching $\theta_1=\theta_K$ in \eqref{eq:harrec} we have a model with balanced demands, in which case we can use the termination condition~\cite[Eq. 19]{Gor90}: $g_{\t}(N)={N+K-1 \choose K-1}\theta^N_1$. The last expression matches $\frac{f^{(K-1)}(\theta_1)}{(K-1)!}$ when $f(x)=x^{N+K-1}$, as in \eqref{eq:recdefdd}.

\section{Hermite-Genocchi}
\label{apx:L2}
We show that if $f(x)=x^{N+K-1}$, then the Hermite-Genocchi theorem \eqref{eq:genocchi} holds under nondistinct demands, {\em i.e.}
\begin{equation*}
  [\theta_1,\mathellipsis,\theta_{K}]x^{N+K-1}=\frac{(N+K-1)!}{N!}\int_{\Delta_{K}} (\theta_1u_1+\mathellipsis+\theta_{K}u_{K})^N d\bs u
\end{equation*}
for arbitrary $\theta_1,\mathellipsis,\theta_{K}$. Using the multinomial theorem~\eqref{eq:multinomial}
\begin{equation*}
  [\theta_1,\mathellipsis,\theta_{K}]x^{N+K-1}={(N+K-1)!}\int_{\Delta_{K}}\sum_{\bs n\geq 0:\atop n=N} \prod_{i=1}^K \frac{\theta_{i}^{n_i}u_{i}^{n_i}}{n_i!} d\bs u
\end{equation*}
Note that the expression can now be simplified by \eqref{eq:dirichlet}, which provides the same identity proved in Appendix \ref{apx:L1}
\begin{equation*}
  [\theta_1,\mathellipsis,\theta_{K}]x^{N+K-1}=\sum_{\bs n\geq 0:\atop n=N} \prod_{i=1}^K \theta_{i}^{n_i}=g_{\t}(N)
\end{equation*}
where we have noted that the summation in the last expression is the specialization of \eqref{eq:bcmp} to single-class models.

\end{document}